\documentclass[11pt]{article}

\usepackage{graphicx}

\usepackage{epsfig,
amssymb,amsmath,amsthm,float
}
\newtheorem{theorem}{Theorem}
\newtheorem{lemma}[theorem]{Lemma}
\newtheorem{corollary}[theorem]{Corollary}

\newtheorem{claim}{Claim}
\newtheorem{definition}[theorem]{Definition}

\def\reals{\mathbb{R}}
\def\R{\reals}

\def\eps{\epsilon}
\def\diam{\mathrm{Diam}}

\def\st{~\mathrm{s.t.}~}
\def\supp{\mathrm{supp}}
\def\conv{\mathrm{conv. hull}}
\def\sym{\bigtriangleup}

\def\cl{\mathrm{cl}}
\def\width{\mathrm{width}}
\def\argmax{\mathrm{argmax}}

\newcommand{\ball}[2]{\mathbb{B}_n(#1,#2)}
\newcommand{\abs}[1]{\left|#1\right|}

\newcommand{\prob}[1]{{\sf Pr}\left(#1\right)}

\newcommand{\vol}[1]{\operatorname{vol}\left(#1\right)}
\newcommand{\set}[1]{\left\{#1\right\}} 
\newcommand{\dpr}[2]{\left\langle #1, #2 \right\rangle }

\def\vol{\mbox{\sf vol}}
\def\E{\mbox{\sf E}}

\def\R{\mbox{\boldmath R}}
\def\eps{\epsilon}

\begin{document}

\title{Thin Partitions: Isoperimetric Inequalities and Sampling Algorithms for some Nonconvex Families}

\author{Karthekeyan Chandrasekaran\thanks{Georgia Tech. Email: {\tt karthe@gatech.edu, dndadush@gatech.edu, vempala@cc.gatech.edu}}
\and Daniel Dadush\footnotemark[1] \and Santosh Vempala\footnotemark[1]
}

\date{}

\maketitle

\begin{abstract}
Star-shaped bodies are an important nonconvex generalization of convex bodies (e.g., linear programming with violations). Here we present an efficient algorithm for sampling a given star-shaped body. The complexity of the algorithm grows polynomially in the dimension and inverse polynomially in the fraction of the volume taken up by the kernel of the star-shaped body. The analysis is based on a new isoperimetric inequality. Our main technical contribution is a tool for proving such inequalities when the domain is not convex. As a consequence, we obtain a polynomial algorithm for computing the volume of such a set as well. In contrast, linear optimization over star-shaped sets is NP-hard.
\end{abstract}

\section{Introduction}
Convexity has been a cornerstone of fundamental polynomial-time algorithms for continuous as well as discrete problems \cite{GLS}. The basic problems of optimization, integration and sampling in $\reals^n$ can be solved efficiently (to arbitrary approximation) for
convex bodies given only by oracles. More precisely,
\begin{itemize}
\item Optimization. Given a convex function $f: \reals^n \rightarrow \reals$, a convex body $K$ specified by a membership oracle and a point in $K$, and $\eps > 0$, find a point $x^* \in K$ s.t. $f(x^*) \le \min_K f(x)$. This can be done using either the Ellipsoid algorithm \cite{YN, GLS}, Vaidya's algorithm \cite{Va} or the random walk approach \cite{BV, LV3}. For important special cases such as linear programming, there are several alternative approaches.
\item Integration.  Given an integrable logconcave function $f: \reals^n \rightarrow \R_+$ as an oracle, and
$\eps > 0$, find $A$ s.t. $(1-\eps)\int f < A < (1+\eps) \int f$. This is done using a Monte Carlo algorithm based on sampling logconcave densities \cite{DFK,AK,LV2,LV3}.
\item Sampling.  Any logconcave density can be sampled efficiently \cite{LV3}. The sampling algorithm is based on a suitable random walk.
\end{itemize}
For the above problems and related applications, both the algorithms and their analyses rely
heavily on the assumption of convexity or its natural extension, logconcavity. For example, for optimization, all the known algorithms use the fact that a local optimum is a global optimum.
Similarly, a key step in the analysis of sampling algorithms is the derivation of isoperimetric inequalities, which are currently known for logconcave functions. Even the proofs of these inequalities (more on this presently) are based on techniques that fundamentally assume convexity. The main motivation of this paper is the following: {\em for what nonconvex bodies/distributions, can the above basic problems be solved efficiently?}

In this paper, we consider a well-studied generalization of convex bodies called {\em star-shaped} bodies. Star-shaped sets come up naturally in many fields, including computational geometry \cite{PS}, integral geometry, mixed integer programming, etc. \cite{COX}. A star-shaped set has at least one point such that every line through the point has a convex intersection with the set. Alternatively, star-shaped sets can be viewed as the union of convex sets, with all the convex sets having a nonempty intersection. The subset of points that can ``see" the full set is called the kernel of the star-shaped set.

A compelling example of a star-shaped set is the ``$k$-out-of-$m$-inequalities" set, i.e., the set of points that satisfy at least $k$ out of a given set of $m$ linear inequalities, with the assumption that there is a feasible solution to all $m$. In this case the kernel is the intersection of all $m$ inequalities. Another interesting special case is that of ``$k$-out-of-$m$-polytopes", where we have $m$ polytopes with a nonempty intersection and feasible points are required to lie in at least $k$ of the $m$ polytopes. These and other special cases have been studied and applied extensively in operations research \cite{RW94, M94, C05}. Not surprisingly, linear optimization over even these special cases is $NP$-hard \cite{LAN07}.

This might suggest that the problems of sampling and integration are also intractable over star-shaped bodies. Indeed convex optimization is reducible to sampling. Our main result (Theorem \ref{SAMPLING1}) is that, to the contrary, {\em star-shaped bodies can be sampled efficiently}, with the complexity growing as a polynomial in $n, 1/\eps, \ln D$ and $1/\eta$, where $n$ is the dimension, $\eps$ is an error parameter denoting distance to the true uniform distribution, $D$ is the diameter of the body and $\eta$ is the fraction of the volume taken up by the kernel; we assume that we are given membership oracles for $S$ as well as for its kernel $K$ and a point $x_0$ so that the unit ball around $x_0$ is contained in $S$. (For the particular cases considered above, these oracles are readily available). The sampling algorithm leads to an efficient algorithm for computing the volume of such a set as well. We note here that linear optimization remains NP-hard even when the kernel takes up most of the volume.

A reader familiar with sampling algorithms for convex bodies will recall that such an analysis crucially uses isoperimetric inequalities. Here we prove isoperimetric inequalities for star-shaped sets (Theorems \ref{ISO1}, \ref{ISO2}). The key technical contribution of this paper is the proof of these inequalities and a new tool we develop for this purpose, which is also of independent interest. We refer to this tool as a {\em thin decomposition} of a set. The other crucial ingredients for efficient sampling (local conductance, coupling, etc...) extend naturally from the convex case to the star-shaped one. Therefore building on this new isoperimetry, we are able to show that the \emph{ball walk} provides an efficient sampler for star-shaped bodies.  

In the rest of this section, we give some context for thin partitions.

The common ingredient of most proofs of isoperimetric inequalities for convex bodies is the localization lemma, introduced by Lov\'{a}sz and Simonovits \cite{LS93}. The approach is based on proof by contradiction. If a certain target inequality is false in $\reals^n$, then there exists an essentially one-dimensional object over which it is still false. The proof is then completed by proving a one-dimensional inequality. This approach has been quite successful for convex bodies and logconcave functions and for proving many other inequalities in convex geometry. These, in turn, have played an essential role in the analysis of algorithms for convex bodies.


However, this approach does not seem to work for nonconvex sets, since the resulting one-dimensional versions  could be nonconvex or nonlogconcave (e.g., for star-shaped bodies, convexity holds along lines that intersect the kernel but is not required along lines that do not intersect the kernel).
To overcome this, we use partitions of $\reals^n$ induced by hyperplanes where each part is ``long" in at most one direction. The overall proof strategy in applying the partition is proof by induction: we combine inequalities on all the parts to derive an inequality for the full set. The advantage of this (as opposed to proof by contradiction) is that a suitably strong inequality does not need to hold for {\em every} part; it suffices to hold for most parts.

\subsection{Preliminaries}
Let $S \subseteq \reals^n$ be a compact body. Define the kernel of $S$ as $K_S := \set{x: x \in S \st \forall y \in S ~ [x,y] \subseteq S}$. We
say $S$ is {\em star-shaped} if $K_S$ is nonempty and let $\eta(S) = \vol(K_S)/\vol(S)$. \\
We denote the $n$-dimensional ball of radius $r$ centered around a point $x$ as $\ball{x}{r}$. The ball walk with step size $\delta$ in a set $S$ is the following Markov process: At a point $x$ in $S$, we pick a uniform
random point in $\ball{x}{\delta}$ and move to the chosen point if it is in $S$ and otherwise stay put.
Let $\pi_S$ denote the uniform measure on $S$ and let $\sigma_m$ denote the measure after $m$ ball walk steps. For two probability distributions $\sigma, \tau$, the {\em total variation} distance is
\[
d_{tv}(\sigma, \tau)=sup_A(\sigma(A)-\tau(A))
\]
\subsection{Results}
We begin with two isoperimetric inequalities for star-shaped bodies, one parametrized using the diameter and the other using the second moment. 
\begin{theorem}\label{ISO1}
Let $S$ be a star-shaped body with diameter $D$ and $\eta(S) > 0$. Then for any measurable partition $(S_1, S_3, S_2)$
of $S$, we have that
\[
\vol(S_3) \geq \frac{\eta(S)}{4D} d(S_1,S_2)\min \set{\vol(S_1), \vol(S_2)}
\]
where $d(S_1,S_2)$ is the minimum distance between a point in $S_1$ and a point in $S_2$.
\end{theorem}

The above theorem is nearly the best possible as shown by a construction in Theorem \ref{thm:iso_upper_bound}. 

\begin{theorem}\label{ISO2}
Let $S$ be a star-shaped body with $\eta(S) > 0$ and $M_S = \E_S[\|X-\mu_S\|^2]$ where $\mu_S$ is the centroid of $S$.
Then for any measurable partition $(S_1, S_3, S_2)$ of $S$, we have that either
\[
\vol(S_3) \geq \frac{\eta(S)}{4} \vol(S)
\]
or
\[
\vol(S_3) \geq \frac{\eta(S)^\frac{3}{2}}{16\sqrt{M_S}} d(S_1,S_2)\min \set{\vol(S_1), \vol(S_2)}
\]
where $d(S_1,S_2)$ is the minimum distance between a point in $S_1$ and a point in $S_2$.
\end{theorem}

Next, we turn to the complexity of sampling. We assume that we have an oracle for the star-shaped body $S$, a lower bound on $\eta$ and an $M$-warm start $\sigma_O$ for the random walk, i.e. an initial distribution on $S$ such that $\forall A\subseteq S$, $\sigma_0(A)\leq M\pi_S(A)$.

\begin{theorem}\label{SAMPLING1}
Let $S$ in $\reals^n$ be a star-shaped body with kernel $K_S$ and $\eta(S) > \eta$, and diameter $D$. Let $\pi_S$ be uniform distribution over $S$ and $\eps > 0$. Given a random point $x_0$ from a distribution $\sigma_0$ such that $\sigma_0$ is an $M$-warm start for $\pi_S$, then there exists an absolute constant $C$ such that, after
\[
m > \frac{Cn^2D^2M^2}{\eta^2\eps^2}\log\frac{2M}{\eps}
\]
steps of the ball walk with $\frac{\eps}{8M\sqrt{n}}$-steps, we have
$d_{TV}(\sigma_m, \pi_S) < \eps$.
\end{theorem}

\begin{theorem}\label{SAMPLING2}
Let $S$ in $\reals^n$ be a star-shaped body with kernel $K_S$ and $\eta (S) >\eta$. Suppose we are given
membership oracles for $K_S$ and $S$ and a point $x_0$ with $\ball{x_0}{1} \subseteq K_S \subseteq S \subseteq
\ball{0}{D}$. Then, for any $\eps > 0$, a nearly random point $x$ from $S$ can be produced using amortized $O^*(n^3/\eta^4\eps^2)$ 
oracle calls with the guarantee that the distribution $\sigma$ of $x$ satisfies $d_{TV}(\sigma,\pi_S) < \eps$.
\end{theorem}
We note that up to the polynomial in $\eta$, this matches the best-known bounds for sampling convex bodies. Due to page restrictions, many of the proofs appear in an appendix.
\section{Thin Decompositions via Bisection}
\begin{definition}\label{def:thin}
Let $S \subseteq \reals^n$. We define $S$ to be a compact body if $S$ is compact, has non-empty interior, and satisfies $\cl(S^\circ) =
S$, where $\cl(S^\circ)$ denotes the closure of the interior of $S$.

Let $S \subseteq \reals^n$ be a compact body. A {\em decomposition} of $S$ is a finite collection ${\cal P} =
\set{P_1,\ldots,P_k}$ of compact bodies such that
\begin{enumerate}
\item $S = \cup_{i=1}^k P_i$
\item $P_i \cap P_j = \partial P_i \cap \partial P_j$, $1 \leq i < j \leq k$
\end{enumerate}
Furthermore, we define a decomposition ${\cal P}$ to be {\em $\eps$-thin} if each $P \in {\cal P}$ is contained
in a cylinder of radius at most $\eps$.
\end{definition}

For completeness, we state without proof the following simple lemma.
\begin{lemma}
\label{lem:decomposition}
Let $S \subseteq \reals^n$ be a compact body.
\begin{enumerate}
\item Let $N$ be a decomposition of $S$, and let $S' \subseteq S$ be a compact body. Then \\ $N' = \set{\cl((P \cap
S)^\circ): P \in N, P \cap S^\circ \neq \emptyset}$ is a valid decomposition of $S'$.
\item Let $N$ be a decomposition of $S$, and let $N'$ be a decomposition of an element $P \in N$. Then $N \cup N'
\setminus P$ is a valid decomposition of $S$.
\end{enumerate}
\end{lemma}

The following simple lemma from \cite{LS93} will be used repeatedly.
\begin{lemma}\label{lem:bisection}
Let $f:\reals^n \rightarrow \reals$ be integrable, $n \ge 2$. Then for any point $z \in \reals^n$, and any $2$-dimensional linear
subspace $S$ of $\reals^n$, there exists a hyperplane $H = \set{x : a^Tx = a^Tz}$, with $a \in S$, inducing halfspaces $H^+,
H^-$, such that it equipartitions $f$, i.e.,
\[
\int_{H^+} f(x) \, dx = \int_{H^-} f(x)\, dx.
\]
\end{lemma}

\begin{theorem}\label{thm:decomposition}
For any integrable function $f:\reals^n \rightarrow \reals$ with $\supp(f) \subseteq S$, $S$ a compact body, and $\int f dx= 0$,
and any $\eps > 0$, there exists an $\eps$-thin decomposition ${\cal P}$ of $S$ such that each part $P \in {\cal P}$ is
obtained by successive half space cuts from $S$ and satisfies $\int_P f dx = 0$.
\end{theorem}
\begin{proof}
Pick $D$ such that $S \subseteq \ball{0}{D}$. Since $S$ is compact we know that $D < \infty$. We start with the initial
decomposition ${\cal P}_0 = \set{S}$ of $S$. From this decomposition, we will inductively build decompositions ${\cal
P}_1,\ldots,{\cal P}_{n-1}$ with the following properties. For each $i$, $0 \leq i \leq n-1$, we have that for all $P
\in {\cal P}_i$:
\begin{enumerate}
\item $P$ is obtained from $S$ via successive half space cuts.
\item $\int_P f dx = 0$
\item $\exists S \subseteq \reals^n$, an $i$-dimensional linear subspace of $\reals^n$ such that the orthogonal projection
of $P$ into $S$ is contained inside of cuboid of side length at most $\delta = \frac{2\eps}{\sqrt{n}}$.
\end{enumerate}
Assuming the above properties, one can easily see that each part in ${\cal P}_{n-1}$ is contained inside a cylinder of
radius $\sqrt{n}\frac{\delta}{2} = \eps$, and hence ${\cal P}_{n-1}$ is an $\eps$-thin decomposition of $S$ compatible
with $f$ as needed. Hence, we need only show how to perform the induction step.

Take $P \in {\cal P}_i$, $0 \leq i \leq n-2$. By assumption, there exists an $i$-dimensional linear subspace $T$
such that $\pi_T(P)$, the orthogonal projection of $P$ into $T$, is contained inside a cuboid of size length at most
$\delta$. Since $i \leq n-2$, we may pick a $2$ dimensional subspace $S$ orthogonal to $T$.

Let $Q = \conv(P)$ and let $\Pi_S$ denote the orthogonal projection map from $\reals^n$ onto $S$. Since $P \subseteq S
\subseteq \ball{0}{D}$ and $\ball{0}{D}$ is convex, we know that $Q \subseteq \ball{0}{D}$. Therefore $\Pi_S(Q) \subseteq \ball{0}{D} \cap S
\Rightarrow \vol_2(Q_S) \leq \pi D^2$. Let $N = \set{ Q }$. We perform the following iterative cutting procedure on $N$.
Take an element $E \in N$. If $\vol_2(\Pi_S(E)) < \delta^2/2$ stop. Otherwise, letting $\mu$ denote the centroid
of $\Pi_S(E)$, we have by lemma $\ref{lem:bisection}$ that there exists $H = \set{x : a^t x = a^t \mu}$, where $a \in
S$, such that $\int_{E \cap H^-} f dx = \int_{E \cap H^+} f dx = 0$. Let $E_1 = E \cap H^-$, $E_2 = E \cap H^+$. Now set
$N \leftarrow N \cup \set{E_1, E_2} \setminus E$. Since we are cutting through the centroid of $\Pi_S(E)$, and
$\Pi_S(E)$ is convex, by Grunbaum's theorem we know that $\vol_2(\Pi_S(E_1)), \vol_2(\Pi_S(E_2)) \leq
(1-\frac{1}{e})\vol_2(\Pi_S(E))$. Therefore, after a number of iterations depending only on $D$, we will have that every
element $E \in N$ has $\vol_2(\Pi_S(E)) < \frac{\delta^2}{2}$.

\begin{claim}Let $E \in N$. There exists $v \in S$, $\|v\|=1$, such that $\width_v(E) = \sup_{x \in E} v^Tx -
\inf_{x \in E} v^Tx \leq \delta$.
\end{claim}
\begin{proof}
Assume not, then note that the diameter of $\Pi_S(E)$ is at least $\delta$. Let
$[u,v]$ be a diameter inducing chord in $\Pi_S(E)$. Let $w,z$ be points on opposite sides of $[u,v]$ such that their
distances from the line $l(u,v)$ are maximum. Then the sum of the distances from $w,z$ to $l(u,v)$ is at least $\delta$
and therefore the area of the quadrilateral induced by these four points is at least $\delta^2/2$, a contradiction.
Hence there exists a direction $v$ such that $\width_v(E) \leq \delta$.
\end{proof}
Note then that the orthogonal projection of $E$ into
the subspace spanned by $v$ and $T$ is contained inside a cuboid of size length at most $\delta$ as needed.

Hence $N$ is now a decomposition of $Q = \conv(P)$, such that each element of $E \in N$ has $i+1$ orthogonal
$\delta$-thin directions. To transform $N$ into a decomposition of $P$, we let $N' = \set{\cl((E \cap P)^\circ): E \in
N, E \cap P^\circ \neq \emptyset}$. By adding $N'$ to ${\cal P}_{i+1}$, we complete the induction step as needed to
prove the theorem.
\end{proof}

\section{Application to Nonconvex Isoperimetry}

The benefit of Theorem \ref{thm:decomposition} is that it will allow us to derive isoperimetric inequalities for
high-dimensional sets without requiring convexity along every line. We show an application to star-shaped bodies.
To gain some intuition, it is useful to understand what the obstructions to isoperimetry in the star-shaped setting
are, as well as to understand why star-shaped bodies have good isoperimetry at all. The following Theorem illustrates what a ``canonical'' bottleneck looks like in the star-shaped setting. 

\begin{figure}[H]
\label{fig:iso-upper-bound}
\begin{center}
\psfig{file=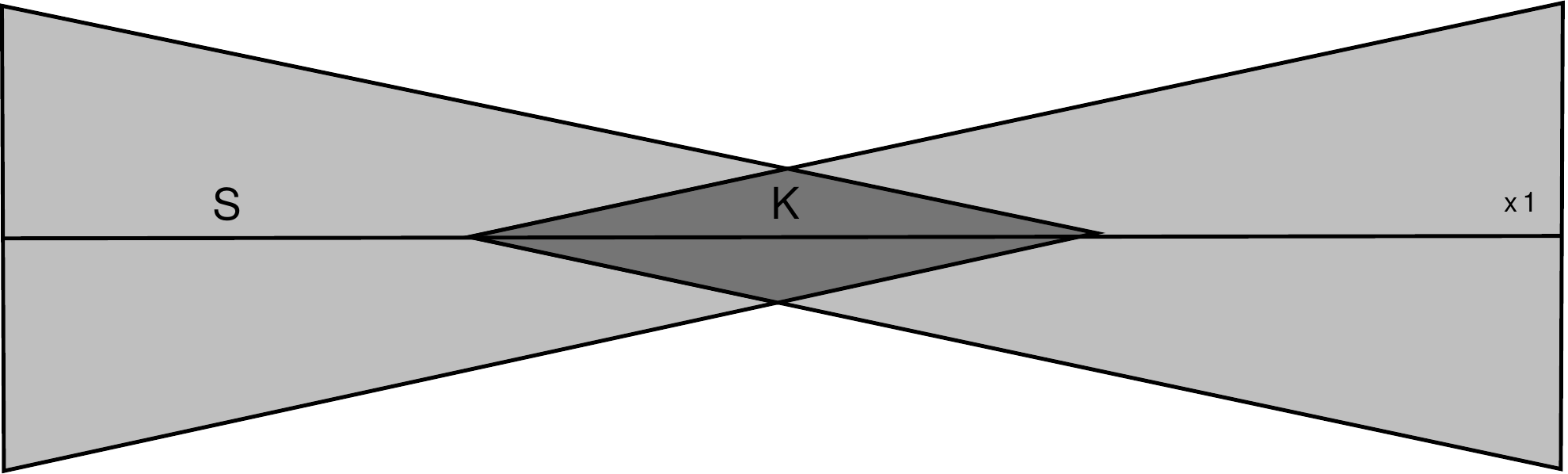,width=2in,height=0.6in}
\end{center}
\caption{Star-Shaped Gluing of 2 Truncated Cones}
\end{figure}

\begin{theorem}
\label{thm:iso_upper_bound}
Let $H_n = \set{x: x \in \reals^n, x_1 = 0}$ and let $(H_n^-,H_n^+)$ denote the halfspaces induced by $H_n$. There exists an
absolute constant $C > 0$, such that for all $\eta > 0$, there exists a sequence of symmetric star-shaped bodies $S_n
\in \reals^n$ centered at $0$ such that for all $n \geq N_\eta$, we have that $\eta(S_n) = \Omega(\eta)$ and 
\[
\vol_{n-1}(H_n \cap S_n) \leq C \left(\frac{\eta \ln(\frac{1}{\eta})}{(1-\eta)\diam(S_n)}\right) \min \set{\vol_n(H^- \cap S_n), \vol_n(H^+ \cap S_n)}
\]
\end{theorem}
\begin{proof}
Our strategy here will be to reduce the above statement to one about $1$ dimensional distributions on the real line. For
each $\eta$, we will construct a candidate sequence $S_n$ of star-shaped bodies which are rotationally symmetric about
the $x_1$ axis. Then by analyzing the cross-sectional distributions of $S_n$ and $K_{S_n}$ along the $x_1$-axis, we will
explicitly construct a $1$ dimensional asymptotic densities $f_\eta, f^K_\eta$ to which the cross-sectional
distributions of $S_n$ and $K_{S_n}$ respectively converge. We will then establish the required isoperimetry and
kernel volume constraints for the sequence $S_n$ and $K_{S_n}$ by direct computation on $f_\eta, f^K_\eta$. 

The geometry of our constructions is simple. As shown in Figure 1 previously, we will take two $n$-dimensional rotational cones with
variance $1$ along the $x_1$ axis, truncate them at their ends removing exactly an $\eta$ fraction of their volume, and
glue them together at the truncation sites. Choose $l_n$ such that $(1-\frac{l_n}{\sqrt{n(n+2)}})^n = \eta$. Since
$l_n \rightarrow \ln(\frac{1}{\eta})$ we may choose $N_\eta$ such that for $n \geq N_\eta$ $2l_n \leq \sqrt{n(n+2)}$.
Now let
\[
S_n = \set{x: \sqrt{\sum_{i=2}^n x_i^2} \leq 1 - \frac{l_n + x}{\sqrt{n(n+2}}, x_1 \in [-l_n, 0]} \bigcup
      \set{x: \sqrt{\sum_{i=2}^n x_i^2} \leq 1 - \frac{l_n - x}{\sqrt{n(n+2}}, x_1 \in [0, l_n]}
\]
From here one can easily verify that the kernel of $S_n$ is
\[
K_{S_n} = \set{x: \sqrt{\sum_{i=2}^n x_i^2} \leq 1 - \frac{l_n - x}{\sqrt{n(n+2}}, x_1 \in [-l_n, 0]} \bigcup
      \set{x: \sqrt{\sum_{i=2}^n x_i^2} \leq 1 - \frac{l_n + x}{\sqrt{n(n+2}}, x_1 \in [0, l_n]}
\]
Next, a simple computation reveals that the cross-sectional distribution of $S_n$ is
\[
f_n(x) = \begin{cases} 
\frac{1}{2(1-\eta)} \frac{\sqrt{n(n+2)}}{n}(1-\frac{l_n + x}{\sqrt{n(n+2)}})^{n-1}:& \quad x \in [-l_n, 0] \\
\frac{1}{2(1-\eta)} \frac{\sqrt{n(n+2)}}{n}(1-\frac{l_n - x}{\sqrt{n(n+2)}})^{n-1}:& \quad x \in [0, l_n] \\
0:& \quad \text{ otherwise}
\end{cases}
\]
Another computation, shows us that the cross-sectional density of $K_{S_n}$ relative to $S_n$ (we normalize by the
volume of $S_n$) is
\[
f^K_n(x) = \begin{cases} 
\frac{1}{2(1-\eta)} \frac{\sqrt{n(n+2)}}{n}(1-\frac{l_n - x}{\sqrt{n(n+2)}})^{n-1}:& \quad x \in [-l_n, 0] \\
\frac{1}{2(1-\eta)} \frac{\sqrt{n(n+2)}}{n}(1-\frac{l_n + x}{\sqrt{n(n+2)}})^{n-1}:& \quad x \in [0, l_n] \\
0:& \quad \text{ otherwise}
\end{cases}
\]
From here, one can easily verify that the sequence $f_n$ converges pointwise to the density function
$f_\eta:\reals \rightarrow \reals^+$ where
\[
f_\eta(x) = \begin{cases} \frac{\eta}{2(1-\eta)} e^{-x} :& \quad x \in [-\ln(\frac{1}{\eta}),0] \\
			  \frac{\eta}{2(1-\eta)} e^x :& \quad x \in [0,\ln(\frac{1}{\eta})] \\
			  0:& \quad \text{otherwise} \end{cases}
\]
Similarly the sequence $f^K_n$ converges pointwise to $f^K_\eta$ where
\[
f^K_\eta(x) = \begin{cases} \frac{\eta}{2(1-\eta)}  e^x :& \quad x \in [-\ln(\frac{1}{\eta}),0] \\
			  \frac{\eta}{2(1-\eta)} e^{-x} :& \quad x \in [0,\ln(\frac{1}{\eta})] \\
			  0:& \quad \text{otherwise} \end{cases}
\]

Notice that $f^K_\eta \leq f_\eta$ and that $f^K_\eta$ is log-concave. We get that
\[
\int_{\reals} f^K_\eta(x) dx = 2\frac{\eta}{2(1-\eta)} \int_0^{\ln(\frac{1}{\eta})} e^{-x} dx = 2\frac{\eta}{2(1-\eta)}
(1-\eta) = \eta
\]
The above computation shows that the volume fraction of the asymptotic kernel is indeed $\eta$ as required. Clearly the
length of the support of $f_\eta$ is $2\ln(\frac{1}{\eta})$. So now we see that the isoperimetric coefficient for
$f_\eta$ is bounded by
\[
\frac{f_\eta(0)}{\min \set{\int_{-\infty}^0 f_\eta(x)dx, \int_{0}^\infty f_\eta(x) dx}} =
\frac{\frac{\eta}{2(1-\eta)}}{\frac{1}{2}} = \frac{\eta}{1-\eta} = 2\frac{\eta\ln(\frac{1}{\eta})}{(1-\eta)|\supp(f_\eta)|}
\]
where $|\supp(f_\eta)|$ denote the length of the support of $f_\eta$.  Clearly the isoperimetry computed above
corresponds asymptotically to that of the partitions $H_n^- \cap S_n,H_n^+ \cap S_n$. To conclude the argument we need
only justify $\diam(S_n) \rightarrow |\supp(f_\eta)|$. As is this is not the case, but this can easily be achieved by
scaling $S_n$ orthogonally to the $x_1$ axis by a factor of $\frac{1}{n^2}$. By doing this, we are collapsing the
sequence $S_n$ onto the $x_1$ axis, without changing the cross sectional distribution along the axis, and hence
asymptotically $\diam(S_n)$ will converge to the to $|\supp(f)_\eta|$ as needed.
\end{proof}

From the above theorem and illustration, we see that the isoperimetry of star-shaped bodies can be strictly worse than in the convex setting where the isoperimetric coefficient is always $\Omega\left(\frac{1}{\diam(S)}\right)$. In particular, from Figure 1, we observe how contrary to the convex setting we can get a V-shaped break in logconcavity of the cross-sectional volume distribution of a star-shaped body. On the other hand,
as we will see later via Lemma \ref{lem:kernel-concave}, the severity of these breaks is strictly controlled by the kernel of $S$. For reference, in Lemma \ref{lem:kernel-concave} we show that the cross-sectional distributions of a star-shaped body satisfy a form of restricted logconcavity with respect to the kernel. The rest of this section will be devoted to proving isoperimetric inequalities for star-shaped sets. In particular, in Theorem \ref{ISO1} we show isoperimetry for star-shaped bodies in terms of the diameter and $\eta$ which in light of Theorem \ref{thm:iso_upper_bound} is optimal within a factor of $O\left(\frac{\ln(\frac{1}{\eta})}{1-\eta}\right)$. 

The next lemma forms the technical core of the isoperimetry proofs for star-shaped sets. Informally, we prove that
for any thin enough hyperplane cut decomposition of a star-shaped set $S$, the parts of the decomposition that intersect
the kernel of $S$ are ``almost'' convex. This will in essence allow us to apply the isoperimetric inequalities developed
for convex sets to the ``almost'' convex pieces from which we will extract the isoperimetric properties of $S$.

\begin{lemma} \label{lem:near-convex-decomposition}
Let $S$ be a star-shaped body with $\eta(S) > 0$, and let $(S_1,S_3,S_2)$ denote a measurable partition of $S$ where
$\vol(S_1),\vol(S_2) > 0$. Then for every $\eps > 0$, there exists a decomposition ${\cal P}$ of $S$ such that
\begin{enumerate}
\item $\forall P \in {\cal P}$, $\vol(S_1 \cap P)\vol(S_2) - \vol(S_1)\vol(S_2 \cap P) = 0$.
\item $\exists N \subseteq {\cal P}$ such that
\begin{enumerate}
\item $\sum_{P \in N} \frac{\vol(P)}{\vol(S)} \geq (1-\eps)\eta(S)$
\item $\forall P \in N$, $P$ is $\eps$-convex, i.e., there exists $P' \subseteq \reals^n$ a convex
body, such that $\vol(P \sym P') \leq \eps \min \set{ \vol(P), \vol(P') }$.
\end{enumerate}
\end{enumerate}
\end{lemma}
First we state and prove a few technical lemmas needed to prove this.
\begin{lemma} \label{lem:tvd}
Let $K_1,K_2 \subseteq \reals^n$ be compact bodies and let $\pi_1, \pi_2$ denote the uniform measures
on $K_1, K_2$ respectively. Then
\[
\vol(K_1 \sym K_2) \leq \eps \min \set{\vol(K_1),\vol(K_2)} \Rightarrow d_{tv}(\pi_1,\pi_2) \leq \eps
\]
\end{lemma}
\begin{proof}
By symmetry, we may assume that $\vol(K_1) \leq \vol(K_2)$. Let $f_1,f_2:\reals^n \rightarrow \reals^+$ denote the associated
densities with respect to $\pi_1$, $\pi_2$. The subsequent computation proves the result:
\begin{align*}
d_{tv}(\pi_1,\pi_2) &= \frac{1}{2} \int_{\reals^n} |f_1(x) - f_2(x)| dx \\
             &= \frac{1}{2}\left(\frac{\vol(K_1 \setminus K_2)}{\vol(K_1)} + \frac{\vol(K_2 \setminus K_1)}{\vol(K_2)}
                + \vol(K_1 \cap K_2)(\frac{1}{\vol(K_1)}-\frac{1}{\vol(K_2)}) \right) \\
             &\leq \frac{1}{2}\left(\frac{\vol(K_1 \setminus K_2) + \vol(K_2 \setminus K_1)}{\vol(K_1)} +
		   \vol(K_1 \cap K_2)(\frac{\vol(K_2) - \vol(K_1)}{\vol(K_1)\vol(K_2)})\right) \\
	     &\leq \frac{1}{2}\left(\frac{\vol(K_1) \eps}{\vol(K_1)} + \vol(K_1)\frac{\eps \vol(K_1)}{\vol(K_1)^2}\right) \\
             &= \frac{1}{2}(\eps + \eps) = \eps \\
\end{align*}
\end{proof}

\begin{lemma} \label{lem:convex-kernel}
For $S \subseteq \reals^n$, $K_S$ is a convex set. Furthermore, if $S$ is compact then $K_S$ is compact.
\end{lemma}
\begin{proof}
If $K_S = \emptyset$ we are done. Therefore assume $K_S \neq \emptyset$, and pick $x,y \in K_S$.  Now
take $z \in [x,y]$. We need to show that $\forall p \in S$, $[z,p] \subseteq S$. Assume not, then there exists $p \in
S$, $q \notin S$, such that $q \in [z,p]$. Since $x,y \in S$ we have that $[x,p],[y,p] \subseteq S$. Furthermore we see that
$q$ is in the interior of the triangle defined by $x,y,p$. Let $l(x,q)$ denote the line through $x,q$. Since $q$ is in
the interior of $\conv \set{x,y,p}$ we must have that $l(x,q)$ intersects the segment $[y,p]$ in some point $r$. But
now note that $r \in S$, since $r \in [y,p]$, and by construction $[x,r] \not\subseteq S$, a contradiction. This proves
the statement.

For the furthermore, we assume that $S$ is compact. To show that $K_S$ is compact, we need only show that
$K_S$ is closed. If $x$ is a limit point of $K_S$, we have a sequence $\set{x_i}_{i=1}^\infty \subseteq
K_S$ converging to $x$. Now to take any point $p \in S$. We see that $[x_i,p] \subseteq S$ for all $i \geq 1$, and
we note that the sequence of line segments $[x_i,p]$ converge to $[x,p]$ as $i \rightarrow \infty$. By compactness of
$S$, we have that the limit segment $[x,p]$ is indeed contained in $S$. Since this holds for all $p$, we see that $x \in
K_S$ are needed.
\end{proof}

\begin{proof} [Proof of Lemma \ref{lem:near-convex-decomposition} (Near Convex Decomposition)]
First we will show that we can find subset $K^r_S \subseteq K_S$ that takes up most and the kernel and that lies deep
inside it, i.e. that $K^r_S + \ball{0}{r} \subseteq K_S$. Formally, let $K^r_S = \set{x: \ball{x}{r} \subseteq K_S}$
where $r > 0$. Let $K_S^\circ$ denote the interior of $K_S$. We note that $K_S^\circ = \cup_{i=1}^\infty
K_S^{\frac{1}{i}}$.  By the continuity of measure, there exists a positive integer $j$, such that for $\eps_0 =
\frac{1}{j}, \vol(K_S^{\eps_0}) \geq (1-\eps)\vol(K_S^\circ)$. Since $K_S$ is convex, we know that $\vol(K_S^\circ) =
\vol(K_S)$ and hence
\[
\frac{\vol(K_S^{\eps_0})}{\vol(S)} \geq \frac{(1-\eps)\vol(K_S)}{\vol(S)} = (1-\eps)\eta(S)
\]

Let $f:\reals^n \rightarrow \reals$ be $f(x)= \vol(S_2) 1_{S_1}(x) - \vol(S_1) 1_{S_2}(x)$ where $1_{S_1},1_{S_2}$ are the
indicator functions for $S_1$ and $S_2$ respectively. We note that $\int_S f = \vol(S_2)\vol(S_1) - \vol(S_1)\vol(S_2) =
0$. By Theorem $\ref{thm:decomposition}$, for every $\eps_1 > 0$, there exists an $\eps_1$-thin decomposition ${\cal
P}_{\eps_1}$ of $S$ such that each part $P\in \cal P$ is obtained by successive half space cuts from $S$ and $\int_Pfdx
= 0$. We note that the condition $\int_Pfdx=0$ immediately implies condition $(1)$ for ${\cal P}_{\eps_1}$. For the time
being we will assume that $\eps_1 < \frac{1}{2} \eps_0$ and determine its exact value later.

Let $N = \set{P: P \in {\cal P}_{\eps_1}, P \cap K_S^{\eps_0} \neq \emptyset}$. Since ${\cal P}_{\eps_1}$ is a
decomposition of $S$, we note that $\cup_{P \in N} P \supseteq K_S^{\eps_0}$ and hence
\[
\frac{\sum_{P \in N} \vol(P)}{\vol(S)} = \frac{\vol(\cup_{P \in N} P)}{\vol(S)} \geq
\frac{\vol(K_S^{\eps_0})}{\vol(S)} \geq (1-\eps)\eta(S)
\]

We will now show that for an appropriately chosen $\eps_1$ every $P \in N$ is $\eps$-convex. Our strategy is as
follows. We analyze a minimal cylinder $C$ of radius $\eps_1$ containing $P$, which exists by our assumption on ${\cal
P}_{\eps_1}$. We will use the fact that $P$ contains a point deep inside the kernel to show that a subcylinder $C'$ of
$C$ is fully contained inside $S$. Lastly we will show that $P' = C' \cap P$ is a convex body whose volume is at least a
$(1-\eps)$ fraction of the volume of $P$. 

Take $P \in N$. Let $C$ be the cylinder of radius at most $\eps_1$ and let $L \subseteq C$ denote the axis of $C$.
Without loss of generality, we may assume that $L$ is a subset of the $x_1$ axis, i.e.
\[
C = \set{x: x \in \reals^n, a \leq x_1 \leq b, \sum_{i=2}^n x_i^2 \leq \eps_1^2} \text{.}
\]

\begin{figure}[t]
\centering
\includegraphics[width=2in]{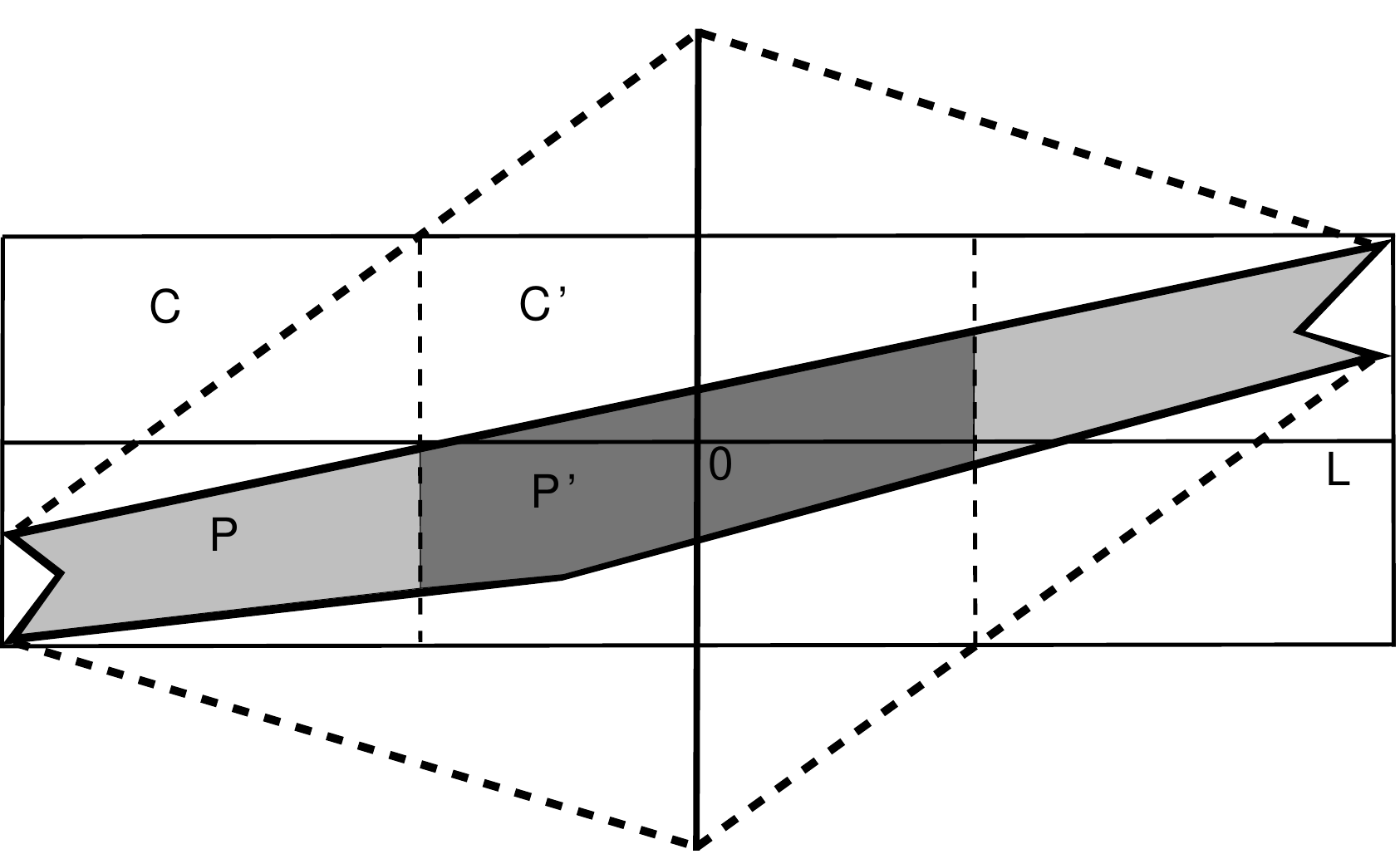}
\caption{P is $\eps$-convex}
\end{figure}


By assumption, we have that $P \cap K_S^{\eps_0} \neq \emptyset$, so pick $c \in P \cap K_S^{\eps_0}$. Since $\eps_1 <
\frac{1}{2} \eps_0$, there exists $d \in L$ such that $\|c-d\| \leq \eps_1 < \frac{1}{2} \eps_0$. Hence
$\ball{c}{\eps_0} \subseteq K_S \Rightarrow \ball{d}{\frac{\eps_0}{2}} \subseteq K_S$. Let $\delta = \frac{1}{2}\eps_0$.
Without loss of generality, we may assume that $d = 0$. Furthermore, by choosing $a,b$ minimal subject to containing
$P$, we may assume there exist points $v,w \in P$ such that $v_1 = a$ and $w_1 = b$. By possibly rotating $C$, we may
assume that $v = (a,r,0,\ldots,0)$ where $0 \leq r \leq \eps_1$. By assumption on $d$, we know that $t =
(0,-\delta,0,\ldots,0) \in K_S$.  Therefore the line segment $[v,t] \subseteq S$. By a simple computation, we see that
$[v,t]$ intersects the $x_1$ axis at $v' = (\frac{\delta}{r + \delta}a,0,\ldots,0)$. Since $0 \in K_S$, we also have
that $[v',0] \in S \Rightarrow v^* = (\frac{\delta}{\eps_1 + \delta}a,0,\ldots,0) \in S$. By symmetric reasoning with
respect to $w$, we have that $w^* = (\frac{\delta}{\eps_1 + \delta}b,0,\ldots,0) \in S$.\\ Now, consider the subcylinder
\[
C' = \set{x: \frac{\delta-\eps_1}{\delta+\eps_1}a \leq x_1 \leq \frac{\delta-\eps_1}{\delta+\eps_1} b, \sum_{i=2}^n
x_i^2 \leq \eps_1^2}
\]
\begin{claim}
$C'\subseteq S$.
\end{claim}
\begin{proof}
Take $x \in C'$. By symmetry we may assume that $x = (e,f,0,\ldots,0)$ where $0 \leq e \leq
\frac{\delta-\eps_1}{\delta+\eps_2}b$ and $0 \leq f \leq \eps_1$. Now examine the line $l(x,w^*)$. A simple computation
reveals that $l(x,w^*)$ intersects the $x_2$ axis at the point $x^* = (0,\frac{\delta b}{\delta b - (\delta +
\eps_1)e}f,0,\ldots)$. Now we note that
\[
\frac{\delta b}{\delta b - (\delta + \eps_1)e}f \leq \frac{\delta b}{\delta b - (\delta + \eps_1)e}\eps_1 \leq
\frac{\delta b}{\delta b - (\delta + \eps_1)\frac{\delta - \eps_1}{\delta + \eps_1}b}\eps_1 =
\frac{\delta b}{\eps_1 b}\eps_1 = \delta
\]
Therefore by assumption on $d$ we know that $x^* \in K_S$. Since $x \in [x^*,w^*]$, we have that $x \in S$ as
needed.
\end{proof}
Now define $P' := P \cap C'$.
\begin{claim}
$P'$ is convex.
\end{claim}
\begin{proof}
To see this note that $P$ is obtained from $S$ via halfspace cuts, i.e.
$P = S \bigcap_{i=1}^m H_i$ where each $H_i$ denotes a halfspace. Now we see that
\[
P' = P \cap C' = S \cap C' \bigcap_{i=1}^m H_i = C' \bigcap_{i=1}^m H_i
\]
since $C' \subseteq S$. Since the intersection of convex sets is convex, we have that $P'$ is convex as needed.
\end{proof}
Now note that $P \sym P' = P \setminus P' = P \setminus C'$. We will now show that for an appropriate choice of
$\eps_1$, depending only on $\delta$ and $\eps$, we have that $P \sym P' \leq \eps \vol(P')$ which will prove that
$P$ is indeed $\eps$-convex. In fact, letting $P_+ = P \cap \set{x: x_1 \geq 0}, P'_+ = P \cap \set{x: x_1 \geq 0}$,
we will prove that
\[
\vol(P_+ \sym P'_+) \leq \eps \vol(P'_+)
\]
By symmetry, the same inequality will follow for the $x_1 \leq 0$ side, and by summing up the two inequalities the
result follows.

Let $S(t) = \set{x: x \in P_+, x_1 = t}$, and $s(t) = \vol_{n-1}(S(t))$. Now let $b' =
\frac{\delta-\eps_1}{\delta+\eps_1}b$ and let $t^* = \argmax_{b' \leq t \leq b} s(t)$.
We have that
\[
\vol_n(P_+ \setminus C') = \int_{b'}^b s(t) dt \leq \int_{b'}^b s(t^*) dt
	= (b-b')s(t^*) = \left(\frac{2\eps_1}{\delta+\eps_1}\right) b s(t^*)
\]
Now by construction the section $S(0) \subseteq K_S$. I claim that $S(0) \subseteq K_{P_+}$. Take $x \in S(0)$ and
$y \in P_+$. Since $x \in K_S$, we have that $[x,y] \subseteq S$. Now $P_+ = S \bigcap_{i=1}^m H_i$. Clearly $x,y \in
P \Rightarrow x, y \in H_i$, for $1 \leq i \leq m$. Furthermore, since each $H_i$ is convex, we have that $[x,y]
\subseteq H_i$. Therefore $[x,y] \subseteq P_+$ as needed. Choose $\alpha \in [0,1]$ such that $(1-\alpha)0 + \alpha t^* =
b'$. Since $S(0) \subseteq K_{P_+}$, we see that
\[
(1-\alpha)S(0) + \alpha S(t^*) \subseteq S(b')
\]
Therefore by the Brunn-Minkowski inequality, we have that
\begin{align*}
s(b') &\geq
\left((1-\alpha)\vol_{n-1}(S(0))^{\frac{1}{n-1}} + \alpha\vol_{n-1}(S(t^*))^{\frac{1}{n-1}}\right)^{n-1} \\
               &\geq \alpha^{n-1} \vol_{n-1}(S(t^*)) \geq \left(\frac{\delta-\eps_1}{\delta+\eps_1}\right)^{n-1} s(t^*)
\end{align*}

Since $P'_+$ is convex we note that $\conv \set{0, S(b')} \subseteq P'_+$ and hence
\[
\vol_n(P'_+) \geq \vol_n\left(\conv \set{0, S(b')}\right) = \frac{1}{n} b's(b') \geq
\frac{1}{n}\left(\frac{\delta-\eps_1}{\delta+\eps_1}\right)^n b s(t^*)
\]

Now by choosing $\eps_1$ small enough such that
\[
\left(\frac{2\eps_1}{\delta+\eps_1}\right) \leq \eps \frac{1}{n}\left(\frac{\delta-\eps_1}{\delta+\eps_1}\right)^n
\]
we get that $\vol(P_+ \sym P'_+) \leq \eps \vol_n(P'_+)$ as needed.

\end{proof}

Using the above lemma, we now prove Theorem \ref{ISO1}.
\begin{proof}[Proof of Theorem \ref{ISO1} (Diameter isoperimetry)]
Let $(S_1,S_3,S_2)$ be a measurable partition of $S$. Without loss of generality we may assume that $\vol(S_1) \leq
\vol(S_2) \Rightarrow \alpha = \frac{\vol(S_1)}{\vol(S_2)} \leq 1$. Note that
\[
\vol(S) \geq \vol(S_1) + \vol(S_2) = \frac{\alpha+1}{\alpha} \vol(S_1) \Rightarrow
\vol(S_1) \leq \frac{\alpha}{\alpha+1} \vol(S)
\]
Let ${\cal P}_{\eps}$ be the decomposition of $S$ with respect to $(S_1,S_3,S_2)$ as defined in Lemma
$\ref{lem:near-convex-decomposition}$ with parameter $\eps$.  Let $N$ denote the set of $\eps$-convex needles. Let $N_+
= \set{P: P \in N, \vol(S_1 \cap P) \geq \frac{1}{2}\frac{\alpha}{\alpha+1} \vol(P)}$ and $N_- = N \setminus N_+$. Since
$N = N_+ \cup N_-$ and by assumption on $N$
\[
\sum_{P \in N} \frac{\vol(P)}{\vol(S)} \geq (1-\eps) \eta(S)
\]
we must have that either
\[
(a) \sum_{P \in N_-} \frac{\vol(P)}{\vol(S)} \geq \frac{1}{2} (1-\eps) \eta(S) \quad \mbox{ or } \quad
(b) \sum_{P \in N_+} \frac{\vol(P)}{\vol(S)} \geq \frac{1}{2} (1-\eps) \eta(S).
\]

Assume first that $(a)$ is true. We will show that $S_1$ and $S_2$ take up a small fraction of most partition parts and
that consequently $S_3$ must take up a large fraction of $S$. Take $P \in N^-$, and let $S_1^P = S_1 \cap P, S_2^P = S_2
\cap P, S_3^P = S_3 \cap P$.  By assumption on ${\cal P}_{\eps}$ we know that $\vol(S_1^P) = \alpha \vol(S_2^P)$.
Therefore we have that
\[
\vol(S_1^P) + \vol(S_2^P) = (\frac{\alpha + 1}{\alpha})\vol(S_1^P) \leq \frac{1}{2}\vol(P)
\]
by assumption on $N^-$. Since $\vol(S_1^P) + \vol(S_2^P) + \vol(S_3^P) = \vol(P)$, we must have that
$\vol(S_3^P) \geq \frac{1}{2}\vol(P)$. Therefore, we have that
\[
\frac{\vol(S_3)}{\vol(S)} \geq \sum_{P \in N^-} \frac{\vol(S_3^P)}{\vol(S)}
\geq \sum_{P \in N^-} \frac{1}{2} \frac{\vol(P)}{\vol(S)} \geq \frac{1}{4} (1-\eps) \eta(S)
\]
Since $\frac{d(S_1,S_2)}{D} \leq 1$, this proves the theorem for case $(a)$.

Now assume that $(a)$ is not true. Then we must have that $(b)$ is true to satisfy our assumption on $N$.  Now take $P
\in N^+$. Our strategy here will be to derive isoperimetry for $P$ using the fact that $P$ is $\eps$-convex. By
approximating the measure of $P$ by that of its convex approximation, we will derive an isoperimetric inequality for $P$
with an additive error depending on $\eps$. Since in this case $S_1$ and $S_2$ take up a lower bounded fraction of $P$,
we will be able to transform the additive error into multiplicative error by making $\eps$ sufficiently small. The
statement will follow as a result.

So let $Q$ be a convex body such that $\vol(P \sym Q) \leq \eps \min \set{\vol(P), \vol(Q)}$. We may assume that $Q
\subseteq \conv(P)$, since otherwise $Q \cap \conv(P)$ is a convex body and strictly closer to $P$. Next, since
$\diam(P) = \diam(\conv(P))$ we have that $\diam(Q) \leq \diam(P) \leq \diam(S) = D$.

Let $\pi_P, \pi_Q$ denote the uniform measures on $P,Q$ respectively.
Let $S_1^Q = S_1^P \cap Q, S_2^Q = S_2^P \cap Q$ and $S_3^Q = Q \setminus (S_1^Q \cup S_2^Q)$. Since
$(S_1^P,S_3^P,S_2^P)$ partition $P$, we note that $S_3^Q = (S_3^P \cap Q) \cup (Q \setminus P)$. Then we have that
$d(S_1^Q,S_2^Q) \geq d(S_1^P,S_2^P) \geq d(S_1,S_2)$. By lemma $\ref{lem:tvd}$ we know that $d_{tv}(\pi_Q,\pi_P) \leq
\eps$ and so we get that
\begin{align*}
\pi_Q(S_1^Q) &= \pi_Q(S_1^P) \geq \pi_P(S_1^P) - \eps \geq \pi_P(S_1^P) - 3\eps \\
\pi_Q(S_2^Q) &= \pi_Q(S_2^P) \geq \pi_P(S_2^P) - \eps \geq \pi_P(S_2^P) - 3\eps \\
\pi_Q(S_3^Q) &\leq \pi_P(S_3^Q) + \eps = \pi_P(S_3^Q \cap P) + \eps \\
             &\leq \pi_Q(S_3^Q \cap P) + 2\eps = \pi_Q(S_3^P) + 2\eps \leq \pi_P(S_3^P) + 3\eps \\
\end{align*}

Since $Q$ is convex, using the isoperimetric inequality proved in \cite{LS93} we have that
\[
\pi_Q(S_3^Q) \geq \frac{d(S_1^Q,S_2^Q)}{\diam(Q)} \min \set{\pi_Q(S_1^Q), \pi_Q(S_2^Q)}
	     \geq \frac{d(S_1,S_2)}{D} \min \set{\pi_Q(S_1^Q), \pi_Q(S_2^Q)}
\]
Now bringing the above inequalities together, we get that
\begin{align*}
& \pi_P(S_3^P) + 3\eps \geq \frac{d(S_1,S_2)}{D} \min \set{\pi_P(S_1^P)-3\eps, \pi_P(S_2^P)-3\eps}  \Rightarrow \\
& \pi_P(S_3^P) \geq \frac{d(S_1,S_2)}{D} \min \set{\pi_P(S_1^P)-3\eps, \pi_P(S_2^P)-3\eps} -3 \eps \Rightarrow \\
& \pi_P(S_3^P) \geq \frac{d(S_1,S_2)}{D}\left(\pi_P(S_1^P)-3\eps\right) -3 \eps \\
\end{align*}
since $\pi_P(S_1^P) \leq \pi_P(S_2^P)$. Now choose $\eps \leq \min \set{ \frac{\eps_0}{12}\frac{d(S_1,S_2)}{D}
\frac{\alpha}{\alpha+1}, \frac{\eps_0}{12} \frac{\alpha}{\alpha+1}}$ where $\eps_0 > 0$. Since $P \in N^+$ we have that
$\vol(S_1^P) \geq \frac{1}{2}\frac{\alpha}{\alpha+1}\vol(P) \Rightarrow \pi_P(S_1^P) \geq \frac{1}{2}
\frac{\alpha}{\alpha+1}$. Hence $3\eps \leq \frac{\eps_0}{4} \frac{\alpha}{\alpha+1} \leq \frac{\eps_0}{2}
\pi_P(S_1^P)$. A simple computation now gives us that
\[
\pi_P(S_3^P) \geq (1-\eps_0)\frac{1}{2} \frac{d(S_1,S_2)}{D}\frac{\alpha}{\alpha+1} \\
\]
Now $\vol(S_3^P) = \pi_P(S_3^P) \vol(P)$, so we see that
\begin{align*}
\vol(S_3) &\geq \sum_{P \in N^+} \vol(S_3^P) \\
          &\geq (1-\eps_0) \frac{1}{2}\frac{d(S_1,S_2)}{D}\frac{\alpha}{\alpha+1} \sum_{P \in N^+} \vol(P) \\
          &\geq (1-\eps_0)\frac{d(S_1,S_2)}{D} \frac{\alpha}{\alpha+1} (1-\eps) \frac{\eta(S)}{2} \vol(S) \\
	  &\geq (1-\eps_0)(1-\eps) \frac{\eta(S)}{4D} d(S_1,S_2) \vol(S_1) \\
	  &= (1-\eps_0)(1-\eps) \frac{\eta(S)}{4D} d(S_1,S_2) \min \set{\vol(S_1), \vol(S_2)} \\
\end{align*}

Finally, letting $\eps_0 \rightarrow 0$ yields the result.
\end{proof}

We prove Theorem \ref{ISO2} following a similar proof strategy as Theorem \ref{ISO1}. We need the following lemma about second moments.
\begin{lemma}\label{lem:var_mix}
Let $f_1,\ldots,f_m:\reals^n \rightarrow \reals^+$ be densities with associated random variables $X_1,\ldots,X_m$ and
centroids $\mu_1,\ldots,\mu_m$ respectively. Let $g = \sum_{i=1}^m p_if_i$ be a mixture of the $f_i$s with
associated random variable $Y$ and centroid $\mu$. Then we have that
\[
\E[\|Y-\mu\|^2] = \sum_{i=1}^m p_i \E[\|X_i-\mu_i\|^2] + \sum_{1 \leq i < j \leq m} p_i p_j\|\mu_i-\mu_j\|^2
\]
\end{lemma}
\begin{proof}
Since $g$ is a mixture, we see that
\[
\mu = \int_{\reals^n} x g(x) dx = \sum_{i=1}^m p_i \int_{\reals^n} x f_i(x) dx = \sum_{i=1}^m p_i \mu_i
\]
Next, note that $\|Y-\mu\|^2 = \dpr{Y-\mu}{Y-\mu}$. The following computation yields the result:
\begin{align*}
\E[\|Y-\mu\|^2] &= \int_{\reals^n} \dpr{x-\mu}{x-\mu}g(x)dx \\
                &= \int_{\reals^n} \dpr{x}{x}g(x)dx - 2\dpr{\mu}{\mu} + \dpr{\mu}{\mu}
                = \sum_{i=1}^m p_i \int_{\reals^n} \dpr{x}{x}f_i(x)dx - \dpr{\mu}{\mu} \\
                &= \sum_{i=1}^m p_i \int_{\reals^n} (\dpr{x}{x}-\dpr{\mu_i}{\mu_i})f_i(x)dx +
		   \sum_{i=1}^m p_i\dpr{\mu_i}{\mu_i} - \dpr{\sum_{i=1}^m p_i\mu_i}{\sum_{j=1}^m p_i \mu_i} \\
                &= \sum_{i=1}^m p_i \E[\|X_i-\mu_i\|^2] + \sum_{i=1}^m p_i\dpr{\mu_i}{\mu_i}
							- \sum_{i=1}^m p_i^2 \dpr{\mu_i}{\mu_i}
					                + \sum_{1 \leq i < j \leq m} 2p_ip_j\dpr{\mu_i}{\mu_j} \\
                &= \sum_{i=1}^m p_i \E[\|X_i-\mu_i\|^2] + \sum_{i=1}^m (1-p_i)p_i\dpr{\mu_i}{\mu_i}
					                - \sum_{1 \leq i < j \leq m} 2p_ip_j\dpr{\mu_i}{\mu_j} \\
                &= \sum_{i=1}^m p_i \E[\|X_i-\mu_i\|^2] + \sum_{i=1}^m \sum_{\substack{j=1 \\ j \neq i}}^m p_ip_j\dpr{\mu_i}{\mu_i}
					                - \sum_{1 \leq i < j \leq m} 2p_ip_j\dpr{\mu_i}{\mu_j} \\
                &= \sum_{i=1}^m p_i \E[\|X_i-\mu_i\|^2] + \sum_{1 \leq i < j \leq m} p_ip_j\dpr{\mu_i}{\mu_i}
							- 2p_ip_j\dpr{\mu_i}{\mu_j} + p_ip_j\dpr{\mu_j}{\mu_j} \\
                &= \sum_{i=1}^m p_i \E[\|X_i-\mu_i\|^2] + \sum_{1 \leq i < j \leq m} p_ip_j \|\mu_i-\mu_j\|^2
\end{align*}
\end{proof}

\begin{proof}[Proof of Theorem \ref{ISO2} (Second moment isoperimetry)]
Let $(S_1,S_3,S_2)$ be the measurable partition of $S$. We may assume $\vol(S_1) \leq \vol(S_2)$ and so $\alpha =
\frac{\vol(S_1}{\vol(S_2)} \leq 1$. Let ${\cal P}_{\eps}$, $N$, $N^+$, $N^-$ be defined as in the proof of Theorem
$\ref{ISO1}$. Again as in Theorem $\ref{ISO1}$ we have the cases $(a)$ and $(b)$. If case $(a)$ occurs, then by the
proof of Theorem $\ref{ISO1}$ we have that
\[
\vol(S_3) \geq \frac{1}{4} (1-\eps) \eta(S) \vol(S)
\]
as needed. So we may assume assume that we are in case $(b)$, i.e that
\[
\sum_{P \in N_+} \frac{\vol(P)}{\vol(S)} \geq \frac{1}{2} (1-\eps) \eta(S)
\]
Now for each $P \in {\cal P}_\eps$, let $\pi_P$ denote the uniform measure on $P$, $\mu_P$ denote
the centroid of $P$, and let $M_P = \E_P[\|X-\mu_P\|^2]$. Now we note that $\pi_S$, the uniform measure
on $S$, is a mixture of the $\pi_P$s, i.e.
\[
\pi_S = \sum_{P \in {\cal P}_\eps} \frac{\vol(P)}{\vol(S)}\pi_P
\]
Therefore by Lemma \ref{lem:var_mix} we have that
\[
M_S = \sum_{P \in {\cal P}_\eps} \frac{\vol(P)}{\vol(S)}M_P + \sum_{\substack{\set{P,Q} \subseteq {\cal P}_\eps \\ P \neq Q}}
\frac{\vol(P)\vol(Q)}{\vol(S)^2}\|\mu_P-\mu_S\|^2 \geq \sum_{P \in N^+} \frac{\vol(P)}{\vol(S)}M_P
\]
Let $V = \sum_{P \in N^+} \vol(P)$. By assumption $V \geq \frac{1}{2} (1-\eps) \eta(S) \vol(S)$, and hence
\[
\sum_{P \in N^+} \frac{\vol(P)}{V} M_P \leq \sum_{P \in N^+} \frac{2}{(1-\eps)\eta(S)} \frac{\vol(P)}{\vol(S)}
\leq \frac{2}{(1-\eps)\eta(S)} M_S
\]
Let $N^* = \set{P: P \in N^+, M_P \leq \frac{4}{(1-\eps)\eta(S)} M_S}$. Since $\sum_{P \in N^+} \frac{\vol(P)}{V} M_P$
is an average of positive numbers by Markov's inequality we must have that
\[
\sum_{P \in N^*} \vol(P) \geq \frac{1}{2}V \geq \frac{1}{4} (1-\eps) \eta(S) \vol(S)
\]

Now take $P \in N^*$. By assumption on $N^* \subseteq N$, there exists $Q$ a convex body such that
$\vol(P \sym Q) \leq \eps \min \set{\vol(P), \vol(Q)}$. In particular, by the construction of Lemma
\ref{lem:near-convex-decomposition} we may assume that $Q \subseteq P$. Let $\bar{Q} = P \setminus Q$,
and let $\pi_{\bar{Q}},\pi_Q$ denote the uniform measures on $\bar{Q},Q$ respectively. We now see that
\[
\pi_P = \frac{\vol(\bar{Q})}{\vol(P)}\pi_{\bar{Q}} + \frac{\vol(Q)}{\vol(P)}\pi_{Q}
\]
As done previously above from Lemma $\ref{lem:var_mix}$ we readily see that
\[
M_P \geq \frac{\vol(Q)}{\vol(P)}M_Q \Rightarrow \frac{\vol(P)}{\vol(Q)} M_P \geq M_Q
\Rightarrow (1+\eps)M_P \geq M_Q
\]

As in the proof of Theorem \ref{ISO1}, let $S_1^Q = S_1^P \cap Q$, $S_2^Q = S_2^P \cap Q$, and $S_3^Q = Q \setminus
(S_1^Q \cup S_2^Q)$. Since $Q$ is a convex set, using the isoperimetric inequality proved in \cite{KLS95} we get that
\begin{align*}
\pi_Q(S_3^Q) &\geq \frac{d(S_1^Q,S_2^Q)}{2\sqrt{M_Q}} \min \set{\pi_Q(S_1^Q),\pi_Q(S_2^Q) }
	    \geq \frac{d(S_1,S_2)}{2\sqrt{(1+\eps)M_P}} \min \set{\pi_Q(S_1^Q),\pi_Q(S_2^Q)} \\
	    &\geq \left(\frac{(1-\eps)\eta(S)}{8(1+\eps)M_S}\right)^{\frac{1}{2}} d(S_1,S_2)
		\min \set{\pi_Q(S_1^Q),\pi_Q(S_2^Q)}
\end{align*}
Using the same analysis as in Theorem $\ref{ISO1}$, the above inequality gives us that
\[
\pi_P(S_3^P) \geq \left(\frac{(1-\eps)\eta(S)}{8(1+\eps)M_S}\right)^{\frac{1}{2}} d(S_1,S_2)
		  \min \set{\pi_P(S_1^P)-3\eps,\pi_P(S_2^P)-3\eps} -3\eps
\]
Now choose
\[
\eps \leq \min \left\{ \left(\frac{\eps_0}{12}\right)\left(\frac{(1-\eps)\eta(S)}{8(1+\eps)M_S}\right)^{\frac{1}{2}}
d(S_1,S_2)\left(\frac{\alpha}{\alpha+1}\right), \left(\frac{\eps_0}{12}\right)\left(\frac{\alpha}{\alpha+1}\right)\right\}
\]
for any $\eps_0 > 0$. By the same analysis as in Theorem $\ref{ISO1}$, we get that
\[
\pi_P(S_3^P) \geq
	(1-\eps_0)\left(\frac{(1-\eps)\eta(S)}{8(1+\eps)M_S}\right)^{\frac{1}{2}} d(S_1,S_2)\frac{\alpha}{\alpha+1}
\]
Using the fact that $\sum_{P \in N^*} \vol(P) \geq \frac{1}{4} (1-\eps) \eta(S) \vol(S)$ we get that
\begin{align*}
\vol(S_3) &\geq \sum_{P \in N^*} \vol(S_3^P) \\
          &\geq (1-\eps_0)\left(\frac{(1-\eps)\eta(S)}{8(1+\eps)M_S}\right)^{\frac{1}{2}}
		d(S_1,S_2)\frac{\alpha}{\alpha+1} \sum_{P \in N^*} \vol(P) \\
	  &\geq \frac{(1-\eps_0)(1-\eps)^{\frac{3}{2}}}{(1+\eps)^{\frac{1}{2}}}
		\left(\frac{\eta(S)^{\frac{3}{2}}}{16\sqrt{M_S}}\right) d(S_1,S_2)\frac{\alpha}{\alpha+1} \vol(S) \\
	  &\geq \frac{(1-\eps_0)(1-\eps)^{\frac{3}{2}}}{(1+\eps)^{\frac{1}{2}}}
		\left(\frac{\eta(S)^{\frac{3}{2}}}{16\sqrt{M_S}}\right) d(S_1,S_2)\vol(S_1) \\
	  &= \frac{(1-\eps_0)(1-\eps)^{\frac{3}{2}}}{(1+\eps)^{\frac{1}{2}}}
		\left(\frac{\eta(S)^{\frac{3}{2}}}{16\sqrt{M_S}}\right) d(S_1,S_2)\min \set{\vol(S_1), \vol(S_2)} \\
\end{align*}

Finally, letting $\eps_0 \rightarrow 0$ yields the result.
\end{proof}

\section{Conductance and mixing time}
\subsection{Local Conductance}
Ball walk on star-shaped bodies could potentially get stuck in points with very small local conductance. Here we prove that most of the points in a star-shaped body have good local conductance. First, we extend a lemma from \cite{KLS97} from convex bodies to star-shaped bodies which leads to the proof of good local conductance. The proof is essentially identical to the case of convex bodies.
\begin{lemma}\label{lem:acrossSurfaceMeasure}
Let $L$ be a measurable subset of the surface of a star-shaped set $S$ in $\reals^n$ and let
\[
S_L:=\{(x,y):x\in S, y\notin S, ||x-y||\leq r, \bar{xy}\cap L\neq \phi\}
\]
Then the $2n-$dimensional measure of $S_L$ is at most
\[
r^{n+1} \vol_{n-1}(L) \frac{\pi_{n-1}}{n+1}
\]
\end{lemma}
\begin{proof} 
By the sub-additivity property of this measure, we have that if $L=L_1+L_2$, then, 
\[
\vol(S_{L_1}\cup S_{L_2})\leq \vol(S_{L_1})+\vol(S_{L_2})
\]
Therefore, it is enough to prove this in the case when $L$ is infinitesimally small. Under this assumption, the measure $\mu(S_L)$ is maximized when the surface of S is a hyperplane in a larger neighborhood of $L$. Then the required measure is at most
\begin{eqnarray*}
\int_{x\in S_L:d(x,L)\leq r} \int_{y\notin S_L:||x-y||\leq r} dy dx &=& \int_{t=0}^r \int_{x\in S_L:d(x,L)=t} \int_{y\notin S_L:||x-y||\leq r} dy dx dt\\
&=& \int_{t=0}^r \int_{x\in S_L:d(x,L)=t} \vol(Cap_r(t)) dx dt\\
\end{eqnarray*}
where $Cap_r(t)$ refers to the intersection of a ball of radius $r$ and a halfspace at distance $t$ from the center of the ball. Hence, 
\begin{eqnarray*}
\mu(S_L)&\leq & \int_{t=0}^r \int_{x\in S_L:d(x,L)=t} \vol(Cap_r(t)) dx dt\\
&=& \int_{x\in L} \int_{t=0}^r t\vol(\mathbb{B}_{n-1}{\sqrt{r^2-t^2}}) dt dx\\
&=& \vol_{n-1}(L) \int_{t=0}^r t\vol(\mathbb{B}_{n-1}{\sqrt{r^2-t^2}}) dt\\
&=& \vol_{n-1}(L) \int_{t=0}^r \pi_{n-1}(\sqrt{r^2-t^2})^{n-1}tdt\\
&=& r^{n+1}\vol_{n-1}(L) \frac{\pi_{n-1}}{n+1}
\end{eqnarray*}
\end{proof}
Recall that the local conductance of a point $x\in S$ is defined as $l(x)=\frac{\vol(\ball{x}{r}\cap S)}{\vol(\ball{0}{r})}$. 
\begin{corollary}\label{corr:avgLocalCond}
Suppose a star-shaped body $S$ contains a unit ball with the origin being inside the kernel. Then the average local conductance $\lambda$ with respect to ball walk steps of radius $r$ is at least
\[
\lambda \geq 1- \frac{r\sqrt{n}}{2}
\]
\end{corollary}
\begin{proof}
Let $S$ and $L$ be as in lemma \ref{lem:acrossSurfaceMeasure}. Then, if we choose $x$ uniformly from $S$ and $u$ uniformly from $\ball{0}{r}$, the probability that $[x,x+u]\cap L\neq \phi$ is at most
\[
\frac{r \vol_{n-1}(L)}{2\sqrt{n}\vol(S)}
\]
Using the whole surface of $S$ to be $L$ we obtain that
\[
\lambda \geq 1-\frac{r \vol_{n-1} (\partial S)}{2\sqrt{n}\vol(S)}
\]
The corollary follows once we lower bound the volume of $S$ in terms of its surface area. This volume can be written as the sum of the volume of cones whose apex is at the origin. Now, since the origin is present within the kernel, it can see every point on the surface of $S$. Hence, for each such cone $C$, $\vol(C)\geq \frac{1}{n}\vol_{n-1}(A_C)$, where $A_C$ denotes the base area of the cone. Summing this up, we get that $\vol(S)\geq\frac{1}{n}\vol_{n-1}(\partial S)$.
\end{proof}
The following lemma is the main result of this section.
\begin{lemma}\label{lem:goodCondBody}
Let $S$ denote a star-shaped body containing the unit ball with the origin present inside the kernel and $S_r$ be defined as
\[
S_r:=\{x\in S:l(x)\geq \frac{3}{4}\}
\]
Then,
\begin{enumerate}
\item $\vol(S_r)\geq (1-2r\sqrt{n})\vol(S)$
\item $\vol(K_{S_r})\geq (1-2r\sqrt{n})\vol(K_S)$
\item $S_r$ is star-shaped.
\end{enumerate}
\end{lemma}
\begin{proof} 
Using Corollary \ref{corr:avgLocalCond}, we get that
\begin{eqnarray*}
\frac{1}{\vol(S)}\int_S (1-l(x))dx &\geq & (1-\frac{r\sqrt{n}}{2})\vol(S)\\
E(1-l(x)) &\leq& \frac{r\sqrt{n}}{2}\\
\prob{\frac{\vol(\ball{x}{r}\cap \bar{S})}{\vol(\ball{0}{r})}\geq \frac{1}{4}} & \leq & 2r\sqrt{n}\\
\prob{l(x)\geq \frac{3}{4}} &\geq& (1-2r\sqrt{n})\\
\frac{\vol(S_r)}{\vol(S)} &\geq& (1-2r\sqrt{n})
\end{eqnarray*}
Applying the same argument as above to the kernel of $S_r$, we obtain the second inequality. The final conclusion is obtained as a consequence of 2.
\end{proof}

\subsection{Coupling}
In this section, we prove that the one-step distributions of points close to each other and having good local conductance overlap by a good fraction. The proof follows the case of convex bodies closely since we are considering only points of good local conductance.
\begin{lemma}\label{lem:coupling}
Let $S$ be a star-shaped body and let $u,v\in S$ such that $|u-v|\leq \frac{tr}{\sqrt{n}}$, $l(u),l(v)\geq l$. Then
\[
d_{TV}(P_u,P_v)\leq 1+t-l
\]
\end{lemma}
\begin{proof} [Proof of Lemma \ref{lem:coupling} (Coupling lemma)]
We prove the inequality in the case when both $\ball{u}{r}$ and $\ball{v}{r}$ are contained within $S$. If not, then the considered case gives an upper bound and hence, we are done.
\begin{eqnarray*}
d_{TV}(P_u,P_v)&=&\frac{1}{2}(\int_{x\in \ball{u}{r}\cup \ball{v}{r}} |P_u(x)-P_v(x)|dx)\\
&=& \frac{1}{2} |P_u(u)-P_v(u)|+\frac{1}{2}|P_u(v)-P_v(v)|+\frac{1}{2}\int_{x\in \ball{u}{r}\cap \ball{v}{r}\backslash u \backslash v}|P_u(x)-P_v(x)|dx\\
&&\ \  + \frac{1}{2}\int_{x\in \ball{u}{r}\cap \ball{v}{r}\backslash u}(P_u(x)-P_v(x))dx+\frac{1}{2}\int_{x\in \ball{v}{r}\cap \ball{u}{r}\backslash v}(P_v(x)-P_u(x))dx\\
&=&\frac{1}{2}(1-l(u))+\frac{1}{2}(1-l(v))+0\\
&&\ \ +\frac{1}{2}\int_{x\in \ball{u}{r}\cap \ball{v}{r}\backslash u}(P_u(x)-P_v(x))dx+\frac{1}{2}\int_{x\in \ball{v}{r}\cap \ball{u}{r}\backslash v}(P_v(x)-P_u(x))dx\\
&=& 1-l +\frac{1}{2\vol(\ball{0}{r})}(2\vol(\ball{0}{r})-2\vol(Cap_r(\frac{tr}{\sqrt{n}})))\\
&\leq& 1-l+t
\end{eqnarray*}
\end{proof}

\subsection{Conductance}
Now, we bound the $s-$conductance of the ball walk on a star-shaped body.
\begin{lemma}\label{lem:sCond}
Let $S\subset \reals^n$ be a star-shaped body with diameter $D$ such that $\eta(S)=\eta$ fraction of its volume is present in its kernel. Then there exists a ball walk radius $r$ such that the s-conductance $\Phi_s$ of ball-walk of radius $r$ is at least $\frac{s\eta}{2^{13}nD}$.
\end{lemma}
\begin{proof} 
Let the radius $r$ of the ball walk step be $s/4\sqrt{n}$. By Lemma \ref{lem:goodCondBody}, this gives us that
\[
\vol(S_r)\geq (1-\frac{s}{2})\vol(S)
\]
Further, the fraction of the volume of the kernel of $S_r$ is
\[
\eta(S_r)=\frac{\vol(K_{S_r})}{\vol(S)}\geq \frac{(1-s/2)\vol(K_S)}{\vol(S)}=(1-\frac{s}{2})\eta
\]
Now, let $A\cup \bar{A}$ be any partition of $S$ into measurable sets with $\vol(A), \vol(\bar{A})>s (\vol(S))$. Define sets
\[
A_1:=\{x\in A\cap S_r:P_x(\bar{A})<\frac{1}{16}\}
\]
\[
A_2:=\{x\in \bar{A}\cap S_r:P_x({A})<\frac{1}{16}\}
\]
\[
A_3:=S_r\backslash A_1\backslash A_2
\]
Now, suppose that $\vol(A_1)\leq \frac{\vol(S)}{3}$. Then the conductance $\phi_s(A,\bar{A})$ is at least
\begin{eqnarray*}
\frac{1}{\min\{\vol(A),\vol(\bar{A})\}} \int_{x\in A\cap S_r\backslash A_1} \frac{1}{16}dx &=& \frac{1}{\min\{\vol(A),\vol(\bar{A})\}} \frac{1}{16} \vol(A\cap S_r\backslash A_1)\\
&\geq & \frac{1}{\min\{\vol(A),\vol(\bar{A})\}} \frac{1}{16} (\vol(A\backslash A_1)-\vol(S\backslash S_r))\\
&\geq & \frac{1}{\min\{\vol(A),\vol(\bar{A})\}} \frac{1}{16} (\frac{2}{3}\vol(A)-\frac{s}{2}\vol(S))\\
&\geq & \frac{1}{\min\{\vol(A),\vol(\bar{A})\}} \frac{1}{16} (\frac{2s}{3}\vol(S)-\frac{s}{2}\vol(S))\\
&\geq & \frac{1}{\min\{\vol(A),\vol(\bar{A})\}} \frac{s}{32} \vol(S)\\
&\geq & \frac{1}{32}
\end{eqnarray*}
and hence we are done. Therefore, we may assume that $\vol(A_1)\geq \frac{\vol(A)}{3}$ and $\vol(A_2)\geq \frac{\vol(\bar{A})}{3}$.\\
Consider $u\in A_1$ and $v\in A_2$. Then,
\[
d_{TV}(P_u,P_v)\geq 1-P_u(\bar{A})-P_v(A)> 1 - \frac{1}{8}
\]
Using Lemma \ref{lem:coupling} (t=1/8), we get $|u-v|\geq \frac{5r}{8\sqrt{n}}$, and hence, $d(A_1,A_2)\geq \frac{5r}{8\sqrt{n}}$. Now, using Theorem \ref{ISO1} on the partition $A_1, A_2, A_3$ of $S_r$, we get that
\begin{eqnarray*}
\Phi_s &\geq& \frac{1}{\min\{\vol(A),\vol(\bar{A})\}}\int_A P_x(\bar{A})dx\\
&\geq& \frac{1}{2}\frac{1}{16}\frac{\vol(A_3)}{\min\{\vol(A),\vol(\bar{A})\}}\\
&\geq& \frac{1}{2^5}\frac{\eta(S_r)d(A_1,A_2)}{4D}\frac{\min\{\vol(A_1),\vol(A_2)\}}{\min\{\vol(A),\vol(\bar{A})\}}\\
&\geq& \frac{1}{2^9}\frac{\eta(5r(1-s/2))}{8\sqrt{n}D}\frac{\min\{\vol(A),\vol(\bar{A})\}}{\min\{\vol(A),\vol(\bar{A})\}}\\
&\geq& \frac{1}{2^{12}}\frac{s(1-s/2)\eta}{nD}\\
&\geq& \frac{1}{2^{13}}\frac{s\eta}{nD}
\end{eqnarray*}
\end{proof}

Using Theorem \ref{ISO2}, one can derive the following bound by proceeding similarly as in the proof of the above lemma.
\begin{lemma}\label{lem:sCondMeanSqDist}
Let $S\subset \reals^n$ be a star-shaped body with diameter $D$ such that $\eta(S)=\eta$ fraction of its volume is present in its kernel. Then for ball walk radius $r=s/4\sqrt{n}$, for any partition $A,\bar{A}$ of $A$ satisfying $\vol(A),\vol(\bar{A})>s(\vol(S))$, the $s$-conductance of $A$ satisfies
\[
\phi_s(A) \geq \frac{\eta}{2^9}\min\left\{ \frac{\vol(S)}{\min\{\vol(A),\vol(\bar{A})\}},\frac{s}{2^7n}\sqrt{\frac{\eta}{M_{S}}} \right\}
\]
\end{lemma}

\subsection{Mixing time}
Let $\pi_S$ denote the uniform distribution over the star-shaped body. Let $\sigma_m$ denote the distribution after $m$-steps of the ball walk on the star-shaped body.
To relate the $s$-conductance to the mixing time, we use the following lemma from \cite{LS93}.
\begin{lemma} \label{lemma:mix}
Let $0 < s \leq 1/2$ and $H_{s} = \sup_{\pi_{S}(A) \leq s} \abs{\sigma_{0}(A) - \pi_{S}(A)}$. Then for every measurable $A \subseteq \reals^{n}$ and every $m \geq 0$,
\[
\abs{\sigma_{m}(A) - \pi_{S}(A)} \leq H_{s} + \frac{H_{s}}{s} \left(1 - \frac{\phi_{s}^{2}}{2}\right)^{m}.
\]
\end{lemma}
\begin{proof} (of Theorem \ref{SAMPLING1})
Suppose $\sigma_0$ be a starting distribution such that there exists $M>0$, $\forall A\subseteq S$, $\sigma_0(A)\leq M\pi_S(A)$.
Now, by definition $H_{s}\leq M\cdot s$. Hence, using Lemma \ref{lemma:mix} and Lemma \ref{lem:sCond},
\[
d_{TV}(\sigma_m,\pi_S)\leq M\cdot s+M\left(1-\frac{s^2\eta^2}{2^{27}n^2D^2}\right)^m \leq M\cdot s+Me^{-ms^2\eta^2/2^{27}n^2D^2}
\]
Replacing $s$ by $\eps/2M$, for $m\geq \frac{2^{29}n^2D^2M^2}{\eta^2\eps^2}\log\frac{2M}{\eps}$, we have $d_{TV}(\sigma_m,\pi_S)\leq \eps$.
\end{proof}

\section{Sampling algorithm}
To obtain a polynomial-time sampling algorithm we make the additional assumption that we are given an oracle to the kernel of the star-shaped body, a point $x_0$ in the kernel and parameters $r,R$ such that $\ball{x_0}{r}$ lies in the kernel and the kernel is contained in a ball of radius $R$.
The sampling algorithm proceeds as follows:
\begin{enumerate}
\item Use the algorithm of \cite{LV1,LV3} to find a transformation of the body $S$ into isotropic position and obtain a random point $x_0$ in $K_S$.
\item Perform $m$ ball-walk steps from $x_0$ on the transformed body $S'$, for each desired random point.
\end{enumerate}
Clearly, by step $1$ above, we have a $\frac{1}{\eta}$-warm start for the ball-walk on $S'$. Now, by Lemma \ref{lem:sCondMeanSqDist}, to obtain a bound on the $s$-conductance, we need an upper bound on the mean square distance $M_{S'}$ of the body $S'$.\\
We next show that when the kernel is isotropic, the body is not far from isotropic. This will bound $M_{S'}$ which along with Lemma \ref{lem:sCondMeanSqDist} and Lemma \ref{lemma:mix} would prove Theorem \ref{SAMPLING2}.

\begin{lemma} \label{lem:kernel-concave} Let $S$ be a star-shaped body and let $K_S$ be the kernel of $S$. For a vector
$v \in \reals^n$, $\|v\|=1$, define
\[
f_S(t) = \vol_{n-1}(\set{x: v^T x = t, x \in \reals^n} \cap S) \quad \text{ and } \quad
f_K(t) = \vol_{n-1}(\set{x: v^T x = t, x \in \reals^n} \cap K_S) \text{,}
\]
the cross-sectional volumes for $S$ and $K_S$ in direction $v$. Then for $x \in \supp(f_K)$ and $y \in \supp(f_S)$
and $\alpha \in [0,1]$ we have that
\[
f_S(\alpha x + (1-\alpha)y)^{\frac{1}{n-1}} \geq \alpha f_K(x)^{\frac{1}{n-1}} + (1-\alpha)f_S(y)^{\frac{1}{n-1}}
\]
Furthermore, for all $x,y \in \reals$ we get that
\[
f_S(\alpha x + (1-\alpha)y) \geq f_K(x)^\alpha f_S(y)^{1-\alpha}
\]
\end{lemma}
\begin{proof} 
Let $S(t),K_S(t)$ denote the cross-sections of $S$ and $K_S$ in direction $v$ at $t$. Since $x \in \supp(f_K), y \in
\supp(f_S)$ we have that $K_S(x),S(y) \neq \emptyset$. Since $K_S(x)$ is part of the kernel we have that
\[
\alpha K_S(x) + (1-\alpha) S(y) \subseteq S(\alpha x + (1-\alpha)y)
\]
Therefore by the Brunn-Minkowski inequality we have that
\begin{align*}
\alpha f_K(x)^{\frac{1}{n-1}} + (1-\alpha)f_S(y)^{\frac{1}{n-1}}
	&\leq \vol_{n-1}(\alpha K_S(x) + (1-\alpha) K_S(y))^{\frac{1}{n-1}} \\
	&\leq \vol_{n-1}(S(\alpha x + (1-\alpha)y))^{\frac{1}{n-1}} = f_S(\alpha x + (1-\alpha)y)^{\frac{1}{n-1}}
\end{align*}
For the furthermore, we note that the statement is trivial if either $f_K(x)=0$ or $f_S(y)=0$. Therefore, we may assume that
$x \in \supp(f_K), y \in \supp(f_S)$. Since the harmonic average is always smaller than the geometric average, the
statement follows directly from our first inequality.
\end{proof}

\begin{lemma}\label{lem:isotropy}
Let $S \subseteq \reals^n$ be a star-shaped body with an isotropic kernel $K_S$ such that $\eta= \vol(K_S)/\vol(S)$. Then, in any direction $v$, for a random point $X$ from $K_S$, we have
\[
\E((v^T X)^2) \le \frac{3328}{\eta^2}.
\]
\end{lemma}
\begin{proof}
Let $v=(1,0,\ldots,0)^T$ w.l.o.g.  Consider the cross-sectional density $f_K$ induced by the kernel along $v$. Since $K_S$
is isotropic, we have that $f_K(0) \geq \frac{1}{8}$ \cite{LV2}.

Next let $f$ be the cross-sectional density of the body $S$ along $v$. It follows that
\[
f(0) \ge \eta f_K(0) \ge \frac{\eta}{8}.
\]

Let $a = \sup \set{x: f(x) < \frac{\eta f_K(0)}{2}, x \leq 0 }$ and $b = \inf \set{x: f(x) < \frac{\eta f_K(0)}{2}, x
\geq 0}$. We claim that $b-a \le \frac{2}{\eta f_K(0)}$. Suppose not, then
\[
\int_a^b f(x)\, dx \ge \frac{\eta f_K(0)}{2} (b-a) > 1.
\]

Now consider a point $x = tb$ for $t > 1$. Then by Lemma \ref{lem:kernel-concave} we have that
\begin{align*}
&f(b) = f\left(\left(1-\frac{1}{t}\right)0 + \frac{1}{t}x\right) \geq
(\eta f_K(0))^{1-\frac{1}{t}}f(x)^{\frac{1}{t}} \Rightarrow  \\
&f(b)^t(\eta f_K(0))^{1-t} \geq f(x) \Rightarrow \eta f_K(0) \left(\frac{f(b)}{\eta f_K(0)}\right)^{t} \geq f(x)
\end{align*}
The same inequality as above can be derived starting from any $b' > b$, and since for every such $b'$ we have
that $f(b') < \frac{\eta f_K(0)}{2}$ by continuity we have that for $t > 1$
\[
f(x) \leq \eta f_K(0) \left(\frac{1}{2}\right)^{t} = \eta f_K(0) e^{-\ln 2 t}
\]
By a symmetric argument, the same bound holds for $x = ta$. Let $p = \int_a^b
f(x)dx$. The following calculation gives the result:
\begin{eqnarray*}
\E((v^Tx)^2) &\le& \int_{-\infty}^\infty x^2f(x)\, dx = p \left(\frac{1}{p} \int_a^b x^2f(x)\right)
	      + \int_{-\infty}^a x^2 f(x)dx + \int_b^\infty x^2f(x)dx \\
	     &\le& p \max \set{a^2, b^2} + \eta f_K(0) \left[\int_{-\infty}^a x^2 e^{-\ln 2(x/a)}\,
dx + \int_{b}^\infty x^2 e^{-\ln 2(x/b)}\, dx\right]\\
	     &=& p \max \set{a^2,b^2} +
			\eta f_K(0)(a^3 + b^3)\left(\frac{1}{2\ln 2} + \frac{1}{(\ln 2)^2} + \frac{1}{(\ln 2)^3}\right) \\
             &\le& (a^2 + b^2) + (2a^2 + 2b^2)6 = 13(a^2 + b^2) \leq 13\frac{4}{\eta^2 f_K(0)^2} \leq \frac{3328}{\eta^2} \\
\end{eqnarray*}
\end{proof}

Proceeding similarly as in the proof of Theorem \ref{SAMPLING1}, one can derive the proof of Theorem \ref{SAMPLING2}.
\begin{proof} [Proof of Theorem \ref{SAMPLING2} (Polynomial time amortized sampling)]
Lemma \ref{lem:isotropy} gives an upper bound on $M_S\leq \frac{2^{12}n}{\eta^2}$. Using Lemma \ref{lem:sCondMeanSqDist}, we get that the $s$-conductance of the ball walk on a star-shaped body $S\subseteq \reals^n$ with the kernel in isotropic position and $\eta=\vol(K_S)/\vol(S)$ satisfies
\[
\Phi_s\geq \frac{\eta^{5/2}s}{2^{21}n^{3/2}}
\]
By the sampling algorithm of \cite{LV1,LV3}, Step 1 of the sampling algorithm takes $O^*(n^4)$ oracle queries.
Since in step 2 of the algorithm, we started the ball-walk on $S$ by choosing a random point from the kernel, and the kernel takes up at least an $\eta$ fraction of the volume of $S$, a random point from it provides an $(1/\eta)$-warm start.
Proceeding similarly as in the proof of Theorem \ref{SAMPLING1}, we get that after $m>\frac{2^{44}n^3}{\eta^4\eps^2}\log\frac{2}{\eta\eps}$ ball walk steps, $d_{TV}(\sigma_m,\pi_S)\leq \eps$. Hence, by performing $m$-steps of the ball-walk for each desired sample, we obtain the desired amortized bound.
\end{proof}

\section{Discussion}
We have presented isoperimetric inequalities and efficient sampling algorithms for star-shaped bodies, based on a new technique called thin partitions. Linear optimization is NP-hard on these bodies, even when the kernel takes up a constant fraction of the body.
\begin{theorem}\label{thm:NP-hard}
Given a star-shaped polytope $S$, it is NP-hard to optimize a linear function over this body for any $\eta(S) < 1$, even if $S$ is well-rounded.
\end{theorem}
This result follows easily from a theorem of Luedtke et al. \cite{LAN07} and we include a proof in the appendix for completeness. Thus, quite unlike convex bodies, linear optimization is NP-hard over star-shaped bodies, but sampling remains tractable. Given the sampling algorithm, we can estimate the volume as follows: given an oracle for the kernel, we can sample from $K_S$ and obtain the volume estimate for $K_S$ using \cite{LV1,LV3}; further, given that $\eta(S)\geq \eta$ we can also estimate $\eta(S)$ using $O(\frac{1}{\eta^2\eps^2})$ samples and output the product of the two as the estimate for volume.

We believe the thin partition approach should be broadly applicable to proving inequalities in convex geometry, especially for inequalities that do not seem reducible to one-dimensional versions (e.g., the KLS hyperplane conjecture \cite{KLS95}).
\bibliographystyle{amsalpha}
\bibliography{ball-walk-ref}

\providecommand{\bysame}{\leavevmode\hbox to3em{\hrulefill}\thinspace}
\providecommand{\MR}{\relax\ifhmode\unskip\space\fi MR }
\providecommand{\MRhref}[2]{%
  \href{http://www.ams.org/mathscinet-getitem?mr=#1}{#2}
}
\providecommand{\href}[2]{#2}
\begin{thebibliography}{DFK91}

\bibitem[AK91]{AK}
D.~Applegate and R.~Kannan, \emph{Sampling and integration of near log-concave
  functions}, STOC '91: Proceedings of the twenty-third annual ACM symposium on
  Theory of computing (New York, NY, USA), ACM, 1991, pp.~156--163.

\bibitem[BV04]{BV}
D.~Bertsimas and S.~Vempala, \emph{Solving convex programs by random walks}, J.
  ACM \textbf{51} (2004), no.~4, 540--556.

\bibitem[Cha05]{C05}
T.M. Chan, \emph{Low-dimensional linear programming with violations}, SIAM J.
  Comput. \textbf{34} (2005), no.~4, 879--893.

\bibitem[Cox73]{COX}
H.S.M. Coxeter, \emph{Regular polytopes}, Dover, 1973.

\bibitem[DFK91]{DFK}
M.E. Dyer, A.M. Frieze, and R.~Kannan, \emph{A random polynomial-time algorithm
  for approximating the volume of convex bodies}, J. ACM \textbf{38} (1991),
  no.~1, 1--17.

\bibitem[GLS88]{GLS}
M.~Gr{\"o}tschel, L.~Lov{\'a}sz, and A.~Schrijver, \emph{Geometric algorithms
  and combinatorial optimization}, Springer, 1988.

\bibitem[KLS95]{KLS95}
R.~Kannan, L.~Lov\'{a}sz, and M.~Simonovits, \emph{Isoperimetric problems for
  convex bodies and a localization lemma}, J. Discr. Comput. Geom. \textbf{13}
  (1995), 541--559.

\bibitem[KLS97]{KLS97}
R.~Kannan, L.~Lov\'{a}sz, and M.~Simonovits, \emph{Random walks and an
  {$O^*(n^5)$} volume algorithm for convex bodies}, Random Structures and
  Algorithms \textbf{11} (1997), 1--50.

\bibitem[LS93]{LS93}
L.~Lov\'{a}sz and M.~Simonovits, \emph{Random walks in a convex body and an
  improved volume algorithm}, Random Structures and Alg., vol.~4, 1993,
  pp.~359--412.

\bibitem[LSN07]{LAN07}
J.~Luedtke, A.~Shabbir, and G.~Nemhauser, \emph{An integer programming approach
  for linear programs with probabilistic constraints}, IPCO '07: Proceedings of
  the 12th international conference on Integer Programming and Combinatorial
  Optimization (Berlin, Heidelberg), Springer-Verlag, 2007, pp.~410--423.

\bibitem[LV06a]{LV3}
L.~Lov\'{a}sz and S.~Vempala, \emph{Hit-and-run from a corner}, SIAM J.
  Computing \textbf{35} (2006), 985--1005.

\bibitem[LV06b]{LV2}
\bysame, \emph{Simulated annealing in convex bodies and an {$O^*(n^4)$} volume
  algorithm}, J. Comput. Syst. Sci. \textbf{72} (2006), no.~2, 392--417.

\bibitem[LV07]{LV1}
\bysame, \emph{The geometry of logconcave functions and sampling algorithms},
  Random Struct. Algorithms \textbf{30} (2007), no.~3, 307--358.

\bibitem[Mat94]{M94}
J.~Matou\v{s}ek, \emph{On geometric optimization with few violated
  constraints}, SCG '94: Proceedings of the tenth annual symposium on
  Computational geometry (New York, NY, USA), ACM, 1994, pp.~312--321.

\bibitem[PS85]{PS}
F.P. Preparata and M.I. Shamos, \emph{Computational geometry: An introduction},
  Springer-Verlag, 1985.

\bibitem[RP94]{RW94}
T.~Roos and W.~Peter, \emph{k-violation linear programming}, Inf. Process.
  Lett. \textbf{52} (1994), no.~2, 109--114.

\bibitem[Vai96]{Va}
P.~M. Vaidya, \emph{A new algorithm for minimizing convex functions over convex
  sets}, Math. Program. \textbf{73} (1996), no.~3, 291--341.

\bibitem[YN76]{YN}
D.B. Yudin and A.S. Nemirovski, \emph{Evaluation of the information complexity
  of mathematical programming problems}, Ekonomika i Matematicheskie Metody 12
  (1976), 128--142.

\end{thebibliography}

\section{Appendix}
\subsection{Optimization over star-shaped body}
Here we prove that optimization over a star-shaped body is NP-hard. In particular, we reduce the clique problem to linear optimization over a star-shaped polyhedron.
\begin{definition}
An instance of CLIQUE($k$) is given by a graph $G(V,E)$. The problem is to decide if there exists clique of size greater
than $k$.
\end{definition}
It is well-known that CLIQUE($k$) is NP-hard.

Note that the NP-hardness of optimization over a star-shaped body does not depend on the fraction of volume of the kernel.
\begin{proof} [Proof of Theorem \ref{thm:NP-hard}]
We reduce solving CLIQUE($k$) to minimizing a linear function over a star-shaped body. Given a CLIQUE($k$) instance
$G(V,E)$, define variables $x\in \reals ^n$. For each edge $e=(i,j)$, define $\psi^{e}\in \reals^n$,
\begin{equation*}
\psi^e_{l}=\left\{
\begin{array}{cc}
1 &\text{if } l=i \text{ or } l=j,\\
0 &\text{otherwise.}
\end{array}\right.
\end{equation*}
For every edge $e$, denote the set of constraints given by $x\geq \psi^e$ as a block constraint.
Consider the following formulation:
\begin{align}
\text{Minimize} f(x)=\sum_{i=1}^n x_i, &\text{ satisfying at least } \binom{k}{2} \text{ block constraints among:}\nonumber\\
\forall e\in E,\ &x\geq \psi^e
\end{align}
Define the feasible polyhedron as $S$.

\begin{claim}
The feasible polyhedron $S$ defined by the above formulation is star-shaped.
\end{claim}
\begin{proof}
First note that any subset of block constraints among the given constraints define a convex body. Thus, the feasible
polyhedron is a union of convex bodies. Further, $x=(1\ 1 \dots 1)^T$ satisfies all the constraints and hence, we have a
non-empty kernel.
\end{proof}

\begin{claim}\label{tweakEta}
By adding new constraints, a new feasible star-shaped polyhedron $S'$ can be created such that $\eta(S')$ is a constant.
\end{claim}
\begin{proof}
Clearly $x_i\geq 1$ $\forall i\in\{1,...,n\}$ is a feasible convex region contained in $K_S$. Therefore, by adding
constraints $x_i\leq a$ $\forall$ $i\in \{1,...,n\}$, for appropriately chosen value of $a(>1)$, one can make $\eta(S')$
a constant. Note that the set $1\leq x_i\leq a$ $\forall$ $i\in \{1,...,n\}$ is still a feasible convex region contained
in $K_S'$.  Specifically, one can choose $a=n$, to see that \[ \eta(S')\geq \left(\frac{n-1}{n}\right)^n \geq
\frac{1}{e} \]
\end{proof}

\begin{claim} \label{cliqueEquivalence}
There exists a clique of size $k$ in $G(V,E)$, if and only if there exists $x\in S$ such that $f(x)\leq k$.
\end{claim}
\begin{proof}
Suppose the graph has a clique $C(V',E')$ of size $k$. Then, consider $x^*\in \reals^n$ such that $x^*_v=1$ $\forall$
$v\in V'$.  Now, for every edge $e=(i,j)\in E'$, $x^*\geq \psi^e$ is satisfied since, $x^*_i=\psi^e_i=1$ and
$x^*_j=\psi^e_j=1$ and $x^*_k\geq 0$ for $k\in V$, $k\neq i,j$. Since $C$ is a clique, $|E'|=\binom{k}{2}$ and
therefore, $\binom{k}{2}$ block constraints will be satisfied which implies that $x^*\in S$. It is straightforward to
check that $f(x^*)=k$.\\ Suppose there exists $\bar{x}\in S$ such that $f(\bar{x})\leq k$. The objective function $f(x)$
can be rewritten as $min_{F\subseteq E:|F|\geq \binom{k}{2}}\{\sum_{i=1}^n\max_{e\in F}\{\psi^e_i\}\}$. Hence, there
exists $\bar{F}\subseteq E$, $|\bar{F}|\geq\binom{k}{2}$, such that the edges in $\bar{F}$ cover at most $k$ vertices.
Clearly, this is possible only when $\bar{F}$ defines a clique of size $k$.
\end{proof}
Suppose there exists an algorithm $A$ to optimize over a star-shaped body $P$ given as an oracle such that $\eta(P)\geq
c$. Now, given an instance of CLIQUE($k$), we formulate the linear programming problem as above. Following claim
\ref{tweakEta} we can find an appropriate value of $a$ and add constraints such that $\eta(S)\geq c$. Further, it is
easy to make $S$ contain a unit ball based on the value of $a$. Finally, the oracle queries can be answered by checking
the number of block constraints satisfied by the point $x$. Hence, we may use $A$ to minimize $f(x)$. Let $z$ be the
objective value obtained by optimizing using $A$. Using claim \ref{cliqueEquivalence}, it is clear that if $z\leq k$,
CLIQUE($k$) is a ``Yes'' instance, otherwise CLIQUE($k$) is a ``No'' instance.
\end{proof}

\end{document}